\documentclass[letterpaper,12pt]{article}
\usepackage[margin = 1.25in]{geometry}
\usepackage{graphicx} 
\usepackage{amssymb}
\usepackage{amsthm}
\usepackage{enumitem}
\usepackage{comment}
\usepackage{mathtools}
\usepackage{mathrsfs}
\theoremstyle{definition}
\newtheorem{example}{Example}
\newtheorem*{example*}{Example}
\newtheorem{definition}{Definition}
\theoremstyle{theorem}
\newtheorem{theorem}{Theorem}
\newtheorem{proposition}{Proposition}

\newtheorem{lemma}{Lemma}

\newtheorem*{axiom*}{Axiom}
\newtheorem{corollary}{Corollary}
\newtheorem{observation}{Observation}
\newtheorem*{observation*}{Observation}
\newtheorem*{claim*}{Claim}
\usepackage{hyperref}

\usepackage[utf8]{inputenc}
\usepackage{natbib}
\usepackage{parskip}
\usepackage{setspace}
\usepackage{subcaption}
\usepackage{xcolor,soul}
\usepackage{pgfplots}

\title{Power and freedom in mechanisms\thanks{We thank Marc Fleurbaey, Stéphane Gonzalez, Sidartha Gordon, Nicolas Motz, Suresh Naidu, Stefan Napel, and Itai Sher, as well as seminar participants at University Paris I, PSE, the University of Saint-Étienne, the University of Bayreuth, Humboldt University, the University of Caen Normandy, the WZB, the NE\&EP online seminar, and the AFSE 2023 conference, for helpful comments and suggestions, including on an earlier version of this paper that circulated under a different title.
Financial support from the Deutsche Forschungsgemeinschaft through CRC TRR 190 (project number 280092119) is gratefully acknowledged.}}
\author{
    Christian Basteck\thanks{WZB Berlin Social Science Center, christian.basteck@wzb.eu}
    \and 
    Ulysse Lojkine\thanks{Sciences Po, Axpo \& CRIS, ulysse.lojkine@sciencespo.fr}
}

\begin{document}

\maketitle

\begin{abstract}

    In a strategy-proof mechanism, the influence of an agent may be measured by the set of outcomes that arise as the agent's (reported) type varies. More specifically, we refer to an agent's influence on her own relevant outcomes as her freedom, and to the influence on outcomes relevant for other agents as her power over others. The framework generalises both the notion of opportunity set from the freedom of choice literature, and established power indices for binary voting. It identifies constrained efficient mechanisms as those that maximise agents' freedom. Applying our framework to the analysis of assignment rules, we provide novel characterisations of the top trading cycles rule and bipolar serial dictatorships in terms of their freedom and power properties.
    
\end{abstract}

\section{Introduction}
Contemporary economic theory addresses issues of freedom and power, but in a scattered fashion. We possess a rich theory of power in voting committees, yet it remains confined to that specific setting. Industrial organisation and bargaining theory also address the concept within their respective frameworks, but no unified theory of power spans these domains. At the same time, there is a literature on freedom of choice \citep{PattanaikXu1990, PattanaikXu1998}, which aims to rank opportunity sets according to the freedom they offer. Again, this is a highly specific framework, which excludes strategic interactions between agents; and it remains largely disconnected from the literature on power. As \citet{DowdingVanHees2008} note, “a growing literature on the formal analysis of freedom emerged largely independent of the literature on power,” which leads them to wish for “the formulation of a general formal framework that enables integration of the analyses of power and freedom.”

A generalisation of the existing theories in several directions thus appears desirable: to extend the analysis of power beyond voting; to extend the analysis of freedom beyond individual choice from a given set; and to integrate both concepts within a common framework. Such a generalisation would also foster dialogue with other disciplines. Entire strands of normative philosophy and political theory—such as \citet{Anderson1999}'s relational egalitarianism or \citet{Pettit2006}'s republicanism—advocate evaluating economic institutions not in terms of welfare, but in terms of freedom and power. So far, economic theory has not been able to contribute much to this discussion, at least once we move beyond voting systems.

This paper takes a first step toward this research agenda, by proposing a tool that applies to a class of mechanisms much broader than those mentioned above—although still constrained by the assumption of strategy-proofness.
To do so, our guiding intuition is based on the \textit{menu} of an agent, the set of possible outcomes she can bring about by varying her (reported)  preferences -- a concept used under this or other names by various authors \citep{hammond1979straightforward, BarberaSonnenscheinZhou1991, gonczarowski2023strategyproofness, KatuscakKittsteiner2025}.

Existing measures of both freedom of choice and voting power can be understood as special cases of this approach. The freedom of choice literature ranks opportunity sets in terms of how much choice they offer. In game-theoretic terms, the same setting can be reinterpreted as a direct mechanism with a single agent who reports a preference list and receives, among the available items, the one she ranked highest. 
An opportunity set can then be interpreted as the menu of a strategy-proof, single-agent mechanism. Regarding voting, it is well known that existing power indices—whether Banzhaf or Shapley-Shubik—can be obtained as the probability that a player is decisive, i.e., able to determine the outcome through their vote. But if the voting game is interpreted as a direct mechanism, then in our terms, an agent is decisive if their menu is \{Yes, No\}, 
and non-decisive if their menu is a singleton, either \{Yes\} or \{No\}.

So far, we have described an agent’s influence over a game’s  outcome interchangeably as power or freedom.
However, in many economic settings, agents may care about a particular dimension of the outcome; for example, their individual allocation of private goods, while being assumed to be indifferent with respect to other dimensions, such as the allocation of private goods to other agents. In such cases, the menu concept can be refined. On the one hand, an agent’s menu in terms of outcome dimensions that are directly relevant for that agent; we refer to this as her freedom.\footnote{This corresponds to the classic notion of a menu, as introduced by \citep{hammond1979straightforward}.} On the other hand, an agent’s menu in terms of outcome dimensions relevant for other agents; we refer to this as her power over others.


As a first result, we identify a direct relationship between power, or more precisely, freedom, and welfare. As both philosophers and voting theorists have long understood, these concepts need not coincide: one may enjoy high welfare without any power—think of the slave of a benevolent master, or a disenfranchised citizen whose preferences happen to align with those of the voting majority. Yet under a weak richness assumption on preferences we show  a connection at another level (Lemma~\ref{lemma:indiv.menu.welfare}): one mechanism gives you more freedom than another (a larger menu in the sense of set inclusion) for any opposing preference profile if and only if it also gives you higher welfare at any profile. At the population level, it means  that a mechanism Pareto-dominates another if and only if it grants more freedom to all agents (Theorem~\ref{thm:menu.pareto.dom}). 
Furthermore, we show that non-bossy mechanisms can be viewed as those that minimize power over others, relative to individual freedom. In fact, for many economic environments, non-bossy mechanisms satisfy a condition of quantitative equality between each agent’s freedom and their power over opponents (Proposition \ref{prop:non-autarkic-and-non-bossy}).

As the main application of our analytical categories, we provide a detailed mapping of freedom and power within the class of assignment mechanisms that are group-strategy-proof and efficient—trading cycles mechanisms characterized by \citet{pycia2017incentive}. In particular, we analyse the trade-off between two properties that may be normatively appealing: first, universal possibility of complete freedom, i.e., granting everyone unrestricted freedom for at least some preference profile; second, minimal bilateral power, i.e., ensuring that no individual can influence another’s outcome beyond what is necessary to achieve efficiency.
Top trading cycle mechanisms have the first property and are characterized by it among group-strategy-proof and efficient rules (Theorem~\ref{thm:ttc-char}); however, they create situations where an agent has maximal bilateral power over another. Bipolar serial dictatorship mechanisms (a class of rules that slightly extends the class of serial dictatorships) have the second property and are characterised by it among group-strategy-proof and efficient rules (Theorem~\ref{thm:bipolar-SD-char}); but this is at the cost of imposing \textit{ex ante} who will have more and who will have less freedom.

\paragraph{Literature review}
An inspiration for this paper is the work of \citet{NapelWidgren2004} and \citet{KurzMayerNapel2021}, who define the power of an agent over an outcome as the sensitivity of that outcome to a ‘‘tremble" in the agent's preferences.\footnote{We thus depart more markedly from concepts of \textit{a priori} \citep{FelsenthalMachover1998, FelsenthalMachover2004} or “formal” \citep{BrahamHoller2005, BrahamHoller2005b, Bervoets2007} power, where the analysis is based exclusively on the game form and not on the agents’ preferences. These debates are discussed in Appendix~\ref{appendix:power-preferences}.} While they confine themselves to voting games, our framework is more general.

Independently from one another, several papers have developed freedom or power measures in interactive contexts that formalise in different ways the common intuition of equilibrium menus.\footnote{This is in contrast with another strand of the literature which conceives power as a property of rules of the game only, abstracting from players' preferences, as in the study of effectivity functions in cooperative game theory \citep{MoulinPeleg1982, AbdouKeiding1991, PetersKaros2018}.} \citet{Rommeswinkel2014} and \citet{Sher2018, Sher2025} develop cardinal freedom measures that capture the flexibility offered to an agent. Our framework is similar in spirit to these papers, but differs from them in that they take a probability distribution over agents’ preferences (or vote profiles) as their starting point, whereas our measure applies to a specific preference profile. Furthermore, whilst \citet{Sher2025} argues that under non-binary, manipulable voting rules, pivotality-based voting power measures over-estimate voters' freedom of choice, we avoid this problem by focussing on strategy-proof mechanisms.

\citet{GoerlachMotz2024} define a cardinal power index in a non-cooperative game by computing the effect of a counterfactual variation in the agent's preferences and like ours, their measure is specific to a given preference profile. But in their paper, the domain of potential preferences for a given player that they consider is the set of preferences of the other players, which makes sense in some situations, but less in others; among others, it prevents measuring power in a game with only one agent, thus disconnecting power from freedom. A second difference is that they measure the effect of such a counterfactual preference change on the outcome through its effect on the actual utility function of the agent, which makes their index dependent not only on the ordinal but also on the cardinal preferences of the agents.

\cite{alva2019strategy} consider a setting with outside options and show that strategy-proof Pareto-improvement is equivalent to an increase in participation. Put differently, they show that menus expand so as to offer more agents something they weakly prefer over an outside option. Instead, we show that menus need to expand in an even stronger, set inclusion sense.\footnote{Neither result implies the other, due to the different invoked richness conditions on preferences.}

Further, our analysis of mechanisms via menus brings us in contact to the mechanism design literature on menu-descriptions of mechanisms, such as \citet{gonczarowski2023strategyproofness}. The last section of the paper makes an intensive use of the literature on assignment rules, including \citet{Papai2000}'s and \citet{pycia2017incentive}'s studies of efficient and group-strategy-proof rules, as well as \citet{BogomolnaiaDebEhlers2005}'s definition of Bi-polar Serial Dictatorship. An alternative proof of \citet{LongVelez2021}'s result also obtains as a corollary of our theorem~\ref{thm:ttc-char}.

\paragraph{Outline of the paper} 
Section~\ref{section:setting} provides the framework for defining menus in strategy-proof mechanisms. In section~\ref{section:freedom}, we propose two ways to compare menus, by set inclusion or by cardinality, and we establish their relation to welfare comparisons. Section~\ref{section:power-over-others} disaggregates between freedom and power over others when each agent cares only about a dimension of the overall outcome, and it provides a sufficient condition for each agent's cardinal freedom to coincide with her cardinal power over others. Finally, Section~\ref{section:assignment} applies these tools to strategy-proof assignment rules, pointing to their shared properties in terms of freedom and its distribution as well as to the specificities of two mechanisms, (Bipolar) Serial Dictatorship and Top Trading Cycles.

\section{Setting}
\label{section:setting}
\subsection*{Outcomes and preferences}
Let $A$ be a set of \emph{outcomes} and $N$ a set of agents. Agents have \emph{preferences} over outcomes, given by a weak order $R_i$ for every $i\in N$;\footnote{I.e., a binary relation that is strongly connected (for all $a,b \in A$ we have $aR_ib$ or $bR_ia$) and transitive (for all $a,b,c\in A$, if $aR_ib$ and $bR_ic$ then $aR_ic$).} let $I_i$ denote the symmetric and $P_i$ denote the asymmetric part of $R_i$.\footnote{I.e., $aI_ib$ iff $aR_ib$ and $bR_ia$ and $aP_ib$ iff $\neg bR_ia$.} 
Further, let $\mathcal{R}_i$ denote a set of possible preference relations for agent $i$ and refer to it as the domain of $i$'s preferences; let $\mathcal{R}=\times_{i\in N}\mathcal{R}_i$ be the corresponding \emph{domain} of possible preference profiles. 
For each $i\in N$, the opposing preference profile is denoted by $R_{-i}=(R_j)_{j\in N\backslash \{i\}}$ and the domain of opposing preference profiles as $\mathcal{R}_{-i}$. For each $G\subseteq N$, define $R_G\in \mathcal{R}_G$ and $R_{-G}\in \mathcal{R}_{-G}$ analogously.

Together the set of outcomes $A$, the set of agents $N$, and the domain of agents' preferences over outcomes $\mathcal{R}$ form an (economic) \emph{environment} $(A,N,\mathcal{R})$.

In many economic environments,  it is reasonable to assume that only some aspects of an outcome are relevant to an agent. For example, in a pure exchange economy where an outcome is an allocation of all goods among all agents, it is commonly assumed that only the own bundle of allocated goods is relevant to an individual agent. 
Formalising this, we say that $a$ and $b$ are equivalent for $i$, denoted $a\equiv_ib$, if for all $R_i\in \mathcal{R}_i$ we have $aI_ib$. For a given environment and agent $i$, this gives rise to equivalence classes of outcomes $A$ under $\equiv_i$ which we denote as $[a]_i=\{b\in A: a\equiv_ib\}$. We refer to $[a]_i$ as an $i$-relevant outcome\footnote{Here we follow \cite{gonczarowski2023strategyproofness}. \cite{barbera2016group} refer to $[a]_i$ as the \emph{consequence} that \emph{alternative} $a$ has for $i$. The role of these equivalence classes, which is to group together outcomes that are not meaningfully distinct from a freedom of choice perspective, can be compared to that of the “attributes" in \cite{NehringPuppe2002, NehringPuppe2008, Sher2018}, but the former are derived from preferences whereas the latter are exogenously defined.} and denote the partition of $A$ into $i$-relevant outcomes as $A_i=\left\{[a]_i|a\in A\right\}$. In a slight abuse of notation we use $R_i\in \mathcal{R}_i$  to also denote $i$'s preferences over $i$-relevant outcomes, i.e., the weak order on $A_i$ induced by $i$'s preferences over outcomes,\footnote{I.e., for all $a,b\in A$, $[a]_iR_i[b]_i$ iff $a R_i b$.} and $\mathcal{R}_i$ to denote the domain of such preferences over  $i$-relevant outcomes.

Last, as a weak richness condition, we may assume that each $i$-relevant outcome may be ranked as most-preferred for at least one of $i$'s possible preferences.

\begin{definition}[Top-rich preference domains]
    Let $(A,N,\mathcal{R})$ be an  environment. We say that $\mathcal{R}_i$, $i\in N$, is top-rich if for every $[a]_i\in A_i$, there exists $R_i\in \mathcal{R}_i$ such that $[a]_iP_i[b]_i$ for all $[a]_i\neq [b]_i\in A_i$. Moreover, $\mathcal{R}$ is top-rich if $\mathcal{R}_i$ is top-rich for every $i\in N$.
\end{definition}

Top-richness is a commonly invoked condition in social choice theory, for example, referred to simply as `rich' by \cite{BarberaSonnenscheinZhou1991, serizawa1995power, barbera1999maximal}, or `minimally rich' by \cite{chatterji2011tops,chatterji2018random};\footnote{The authors restrict attention to strict preferences, so that no two outcomes are equivalent and $A_i$ may be identified with $A$ itself; see also Example \ref{exmp.no.equiv}.} in the context of (object) assignment rules, \cite{hu2025local} refer to it as `top-one richness'. Moreover, many other common domain assumptions imply top richness:

\begin{example}
Consider a finite set of agents $N$ and a set of objects $O$ with $|N|=|O|$. Each agent should receive one object, so that the set of outcomes $A$ is the set of all bijections $\mu:N\rightarrow O$ while an agent's $i$-relevant outcome can be identified with the object they receive under $\mu$. Let the domain of agents preferences over objects be the domain of all strict orderings. Then, as any object is ranked uniquely at the top for some of $i$'s possible preferences, the domain is top-rich. 
\end{example}

\begin{example}\label{exmp.no.equiv}
Consider a set of voters $N$ who have to choose an alternative $a\in A$. Agents have single-peaked preferences over $A$. Since no agent $i$ is indifferent between any two alternatives for \emph{all} of her possible preferences, the equivalence classes $[a]_i$ are all singletons and $A_i$ may be identified with $A$ itself. As any alternative may be uniquely top-ranked by any agent, the domain of (single-peaked) preferences is top-rich.
\end{example}

\subsection*{Mechanisms, Social Choice Rules, and Strategy-Proofness}
A mechanism allows agents in a given economic environment to each choose a strategy and thereby jointly determine an outcome. We will focus on mechanisms where each agent is sure to have a (weakly) dominant strategy. Thus, w.l.o.g., we consider direct revelation mechanisms where each agent reports her preferences, finds it optimal to do so truthfully, and the mechanism maps preference profiles to outcomes -- formally, a strategy-proof social choice rule.

\begin{definition}[(Strategy-proof) Mechanism]
    Given an environment $(A,N,\mathcal{R})$, a (direct revelation) mechanism is a singleton valued social choice rule, i.e., a function:
    \begin{align*}
        \phi: \; & \mathcal{R}_1 \times ... \times \mathcal{R}_n \rightarrow A \\ 
        & (R_1, ..., R_n) \mapsto a.
    \end{align*}
    A mechanism  $\phi$ is \emph{strategy-proof}, if for all $i \in N$, $R \in \mathcal{R} $, $R'_i \in \mathcal{R}_i$ we have
    \begin{equation*}
        \phi(R)\; R_i \; \phi(R'_i, R_{-i}).
    \end{equation*}
\end{definition}
\medskip


\begin{definition}[Menus]
\label{def:menus}
    Any (direct revelation) mechanism $\phi$, together with a given opposing profile of reported preferences $R_{-i}\in \mathcal{R}_{-i}$, gives rise to an effective choice set of agent~$i$:\[\mathcal{M}_i^\phi(R_{-i}):=\{\phi(R_i,R_{-i})|R_i\in \mathcal{R}_i\}\subseteq A.\]
Distinguishing only between different $i$-relevant outcomes, we correspondingly define \[
\mathcal{M}_{i|i}^\phi(R_{-i}):=\{[\phi(R_i,R_{-i})]_i|R_i\in \mathcal{R}_i\}\subseteq A_i,\] i.e., the canonical projection of $\mathcal{M}_i^\phi(R_{-i})$ onto $A_i$, which we will call the ($i$-relevant) \emph{menu} of $i$ (given $R_{-i}$ and $\phi$).\footnote{Here we follow \cite{gonczarowski2023strategyproofness}. Other authors have referred to these sets 
as option sets \citep{BarberaSonnenscheinZhou1991}, proper budget sets \citep{LeshnoLo2021}, feasible sets \citep{KatuscakKittsteiner2025}.} 
\end{definition}

An agent's menu reflects the influence that they have on their own relevant outcome.  
For example a `dictatorship', understood as a mechanism on the universal domain of strict preferences, is described by the property that one agent's menu allows them to choose among all possible outcomes, i.e., corresponds to $A$, while all other agents have no freedom of choice -- their menus are always a singleton.

Similarly, in a binary voting game with outcomes Yes or No, a given agent's menu consists a singleton, \{Yes\} or \{No\}, if her vote has no influence on the outcome given the profile of others' preferences, whereas her menu is \{Yes, No\} if there is no majority without her vote, i.e., if she is decisive. As developed in Appendix~\ref{appendix:power-preferences}, we believe that such an understanding of the ability of a player to influence an outcome overcomes the dichotomy drawn by \citet{FelsenthalMachover1998, FelsenthalMachover2004} between \textit{a priori} and actual power.

\begin{definition}[Freedom of choice]
    Given an environment $(A,N,\mathcal{R})$, a mechanism $\phi$ grants (weakly) more \emph{freedom (of choice)} to agent $i\in N$ than mechanism $\varphi$, if $\mathcal{M}_{i|i}^\phi(R_{-i})\supseteq \mathcal{M}_{i|i}^\varphi(R_{-i})$ for all $R_{-i}\in \mathcal{R}_{-i}$; it grants $i\in N$ strictly more freedom if there is some $R_{-i}\in \mathcal{R}_{-i}$ for which the inclusion is strict. A mechanism $\phi$ \emph{maximises individuals' freedom}, if there is no mechanism $\varphi$ that awards each $i\in N$ weakly more freedom and strictly more to some agent $j\in N$.
\end{definition}

While we will be interested in menus as a measure of freedom that an agent enjoys under different mechanisms (or different agents within the same mechanism\footnote{Observe that a `dictatorship' also maximises freedom, in that there is no alternative mechanism that grants more freedom to \emph{all} agents, including the `dictator'. This motivates the development of measures by which we can compare freedom between agents.}), they also give rise to a well-known description of strategy-proofness:

\begin{theorem}[\cite{hammond1979straightforward}]\label{Hammond79}
    A mechanism $\phi$ is strategy proof iff at each $R\in \mathcal{R}$, every agent is ensured an $R_i$-optimal choice from $\mathcal{M}_{i|i}^\phi(R_{-i})$, i.e., iff \[\forall R\in \mathcal{R}, x\in \mathcal{M}_{i|i}^\phi(R_{-i}): \quad [\phi(R_i,R_{-i})]_i \, R_i \, x.\]
\end{theorem}

Last, we will also be interested in comparing different mechanisms by their welfare, in particular in terms of Pareto-dominance.

\begin{definition}[Pareto-domination]
     Given an environment $(A,N,\mathcal{R})$, a mechanism $\phi$ Pareto-dominates another mechanism $\varphi$ iff 
     \[\forall R\in \mathcal{R}, i\in N: \quad [\phi(R)]_i R_i  [\varphi(R)]_i,\] where the preference is strict for at least one $R$ and $i$.
\end{definition}

\begin{definition}[(Constrained) Pareto-efficiency]
     Given an environment $(A,N,\mathcal{R})$, a mechanism $\phi$ is Pareto-efficient if there is no mechanism $\varphi$ that Pareto-dominates it. A strategy-proof mechanism $\phi$ is constrained (Pareto-)efficient if there is no strategy-proof mechanism $\varphi$ that Pareto-dominates it.
\end{definition}

\section{Freedom}
\label{section:freedom}
Freedom and welfare are distinct concepts. An agent can be satisfied even when she has no freedom of choice in our sense -- think of a simple majority vote where there is a consensus on the outcome. Each individual agent is satisfied with the outcome, while at the same time having no individual possibility to change it if she would like to.\footnote{In such a situation, we can say that the individual agent is \textit{lucky} in the sense of \citet{Barry1980-1}, i.e. successful not because of herself, but thanks to others.} But this example already suggests that such a happy coincidence cannot be systematic.

Indeed, our next result shows a connection between freedom and welfare. Since a strategy-proof mechanism chooses optimally from each agents' menu, a strategy-proof mechanism that offers agents larger menus than an alternative mechanism (in a set inclusion sense) cannot lead to worse outcomes -- and under a mild richness condition on preferences, the former will Pareto-dominate the latter. Perhaps surprisingly, the converse also holds: the \emph{only} possible way to improve upon a strategy-proof mechanism is to ensure that all agents have weakly larger menus at all possible preference profiles.

\begin{lemma}\label{lemma:indiv.menu.welfare}
    Consider  an environment $(A,N,\mathcal{R})$ and $i\in N$ where $\mathcal{R}_i$ is top-rich, and two strategy-proof mechanisms $\phi$ and $\varphi$. The following statements are equivalent:
    \begin{itemize}
        \item For all $R\in \mathcal{R}$, we have $\phi(R)\, R_i \, \varphi(R)$.
        \item $\phi$ grants more freedom to $i$ than $\varphi$, i.e., for all $R_{-i}\in \mathcal{R}_{-i}$ we have $\mathcal{M}_{i|i}^\phi(R_{-i})\supseteq\mathcal{M}_{i|i}^\varphi(R_{-i})$.
    \end{itemize}
    \end{lemma}

\begin{proof}
Suppose that for all $R\in \mathcal{R}$ we have $\mathcal{M}_{i|i}^\phi(R_{-i})\supseteq\mathcal{M}_{i|i}^\varphi(R_{-i})$. By strategy-proofness (Theorem \ref{Hammond79}), $[\phi(R)]_i R_i x$ for all $x \in \mathcal{M}_{i|i}^\varphi(R_{-i})\subseteq \mathcal{M}_{i|i}^\phi(R_{-i})$, thus in particular for $x=[\varphi(R)]_i.$ 

\noindent For the other direction, suppose that for all $R\in \mathcal{R}$ we have $\phi(R)\, R_i \, \varphi(R)$ and, towards a contradiction, assume there exist $R^*\in \mathcal{R}$ and $z\in \mathcal{M}_{i|i}^\varphi(R^*_{-i})\backslash\mathcal{M}_{i|i}^\phi(R^*_{-i})$. By top-richness of $\mathcal{R}$, there exists some $R_i^z\in \mathcal{R}_i$ such that $z$ is the unique $R_i^z$-most preferred $i$-relevant outcome, i.e., $zP^z_i x$ for all $x\in A_i\backslash \{z\}$. But then $[\varphi(R_i^z,R^*_{-i})]_i=zP_i^z x$ for all $x\in \mathcal{M}_{i|i}^{\phi}(R^*_{-i})\subseteq A_i\backslash \{z\}$, in particular $[\varphi(R_i^z,R^*_{-i})]_i P_i^z [\phi(R_i^z,R^*_{-i})]_i$ -- a contradiction. 
\end{proof}

\begin{theorem}\label{thm:menu.pareto.dom}
    Consider  an environment $(A,N,\mathcal{R})$ where $\mathcal{R}$ is top-rich, and two strategy-proof mechanisms $\phi$ and $\varphi$. The following statements are equivalent:
    \begin{itemize}
        \item $\phi$ Pareto-dominates $\varphi$.
        \item $\phi$ grants weakly more freedom than $\varphi$ to each $i\in N$ and strictly more to some $i^*\in N$. 
    \end{itemize}
    \end{theorem}

\begin{proof}
Suppose that for all $i\in N$, $R\in \mathcal{R}$ we have $\mathcal{M}_{i|i}^\phi(R_{-i})\supseteq\mathcal{M}_{i|i}^\varphi(R_{-i})$. Thus, by Lemma \ref{lemma:indiv.menu.welfare},  $[\phi(R)]_i R_i  [\varphi(R)]_i$ for all $i\in N$, $R\in \mathcal{R}$. Further, suppose there exists $i\in N$, $R^*_{-i}\in \mathcal{R}_{-i}$ and $z\in \mathcal{M}_{i|i}^\phi(R^*_{-i})\backslash\mathcal{M}_{i|i}^\varphi(R^*_{-i})$. By top-richness of $\mathcal{R}$, there exists some $R_i^z\in \mathcal{R}_i$ such that $z$ is the unique $R_i^z$-most preferred $i$-relevant outcome, i.e., $zP^z_i x$ for all $x\in A_i\backslash \{z\}$. But then $[\phi(R_i^z,R^*_{-i})]_i=zP_i^z x$ for all $x\in \mathcal{M}_{i|i}^{\varphi}(R^*_{-i})\subseteq A_i\backslash \{z\}$, in particular $[\phi(R_i^z,R^*_{-i})]_i P_i^z [\varphi(R_i^z,R^*_{-i})]_i$. Thus $\phi$ Pareto-dominates $\varphi$.

For the other direction, suppose that $\phi$ Pareto-dominates $\varphi$. Then, by Lemma \ref{lemma:indiv.menu.welfare}, $\mathcal{M}_{i|i}^\phi(R_{-i})\supseteq\mathcal{M}_{i|i}^\varphi(R_{-i})$ for all $i\in N$, $R\in \mathcal{R}$. Moreover, there exists $i\in N$, $R^*\in \mathcal{R}$ such that $[\phi(R^*)]_i P_i^* [\varphi(R^*)]_i$. But then $[\phi(R^*)]_i\notin \mathcal{M}_{i|i}^\varphi(R^*_{-i})$, thus $\mathcal{M}_{i|i}^\phi(R^*_{-i})\supsetneq\mathcal{M}_{i|i}^\varphi(R^*_{-i})$.
\end{proof}

\begin{example}\label{exmp.imposition}
    Consider a set of agents with preferences over two alternatives $a$ and $b$, and compare the following two strategy-proof mechanisms: $\phi$ imposes $a$ irrespective of agents preferences, whereas $\psi$ selects alternative $b$ if all players weakly prefer $b$ to $a$, and $a$ otherwise. $\psi$ Pareto-dominates $\phi$, and offers everyone weakly larger menus.
\end{example}

In  particular, this implies that a strategy-proof mechanism can be (constrained) Pareto-efficient only if further enlarging the freedom of some agent will reduce the freedom of some other agent.

\begin{corollary}\label{cor:constr.eff.max.free}
     Consider  an environment $(A,N,\mathcal{R})$ where $\mathcal{R}$ is top-rich. Then a strategy-proof mechanisms $\phi$ is constrained Pareto-efficient if and only if it maximises individuals' freedom.
\end{corollary}

\paragraph{Remark} To see the role of the assumption that $\mathcal{R}$ is top-rich, consider a set of possible outcomes $A = \{a, b, c\}$, an agent $i\in N$ with a (non top-rich) preference domain $\mathcal{R}_i = \{ R^1_i, R^2_i \}$ where $R^1_i: aP_i^1bP_i^1c$ and $R^2_i: bP_i^2aP_i^2c$ and other agents who are indifferent between all outcomes. Further, consider two mechanisms $\phi$ and $\psi$ such that $\phi(R^1_i,R_{-i}) = a$, $\phi(R^2_i,R_{-i}) = b$, and $\psi(R^1_i,R_{-i}) = \psi(R^2_i,R_{-i}) = c$ for all $R_{-i}$. Then $\phi$ Pareto dominates $\psi$ but $\mathcal{M}_{i|i}^\psi(R_{-i}) = \{ c \}$ is not a subset of $\mathcal{M}_{i|i}^\phi(R_{-i}) = \{ a,b \}$.

\subsection*{Cardinal index}\label{subsec.carinal}
While comparing freedom of choice under different mechanisms based on set inclusion of their menus is very compelling whenever such comparisons can be made, a drawback of this conservative approach lies in the fact that oftentimes agents' menus will not be directly comparable in this way. 

To construct a more fine-grained measure of freedom of choice, we adapt the canonical cardinal ranking of opportunity sets to our framework. For any mechanism $\phi$ in an environment whith finitely many outcomes $A$,  we refer to $|\mathcal{M}^{\phi}_{i|i}(R_{-i})|$ as the \textit{cardinal index of freedom} (of $i$, given $R_{-i}$). The index counts the number of different $i$-relevant outcomes in $i$'s menu. Comparisons based on this cardinal index induce an ordering over menus that is finer than set inclusion.

\citet{PattanaikXu1990} have proposed an axiomatisation for the cardinal ordering of opportunity sets, based on three properties — indifference between no-choice situations, independence, and strict monotonicity. The same axioms can be used to characterise our cardinal index. Where we innovate is by applying the cardinal measure to an endogenously defined menu rather than an exogenous opportunity set, and further by applying it to $\mathcal{M}^{\phi}_{i|i}(R_{-i})$ rather than to $\mathcal{M}^{\phi}_{i}(R_{-i})$, i.e., count only elements that are distinct from $i$'s welfare perspective. Thus, we partially address a criticism of the simple cardinal approach pointed out by \cite{PattanaikXu1990} themselves: enlarging the menu of an agent by an alternative that does not differ in a relevant sense would force the ranking of choice sets to declare that the enlarged set offers strictly more freedom.\footnote{\cite{PattanaikXu1990} provide an example where an agent is offered different modes of transportation in different option sets and the option `red car' is added to a set that already includes the option `blue car'.} 
To address this weakness, \citet{PattanaikXu2000} introduce exogenous measures of `similarity' of alternatives, i.e., with similarity defined independently of agents preferences. Instead, we declare some outcomes $a, b\in A$ similar, or rather identical from the point of view of agent $i$, whenever $[a]_i=[b]_i$, i.e., whenever an agent is indifferent between them at every possible preference relation, thus endogenising the notion of similarity by deriving it from agents' preferences.

Theorem \ref{thm:menu.pareto.dom} pointed out a close relation between freedom of choice as described by menus and welfare in the Paretian sense. The introduction of the cardinal index allows us to extend that result in a quantitative direction by pointing out the relation between an agent's cardinal freedom index and the ranking of the outcome in her preferences when they are randomly drawn.

\begin{definition}[Rank]
\label{def:rank}
    Consider $\phi$ is a mechanism, $R$ a preference profile for all agents, $i$ an agent such that $A_i$, the set of $i$-relevant outcomes is finite. Then $\rho_i^\phi(R)$ is the rank of the $i$-relevant outcome in $i$'s preferences:
    \begin{equation*}
        \rho_{i}^\phi(R) = | \{ x \in A_i | x R_i [\phi(R)]_i\} |
    \end{equation*}
\end{definition}

We are now equipped to describe the relation that exists between an agent's freedom index and the rank of the outcome in her preferences, when the latter are random and the outcome space is finite.

\begin{proposition}
\label{prop:delta-rank}
Consider two environments $(A,N,\mathcal{R})$, $(A,N',\mathcal{R}')$, two agents $i \leq N$, $j\leq N'$ such that $\mathcal{R}_i$, $\mathcal{R}'_j$ are the sets of strict orderings on $A_i$ and $A_j$ with $|A_i|=|A_j|$, and two strategy-proof mechanisms $\phi$ (on $(A,N,\mathcal{R})$) and $\psi$ (on $(A,N',\mathcal{R}')$). Suppose $R_{-i}$, $R'_{-j}$ are given while $R_i$ and $R'_j$ are both uniformly random in $\mathcal{R}_i$. Then:
\begin{equation*}
    |\mathcal{M}_{i|i}^\phi(R_{-i})| \geq |\mathcal{M}_{j|j}^\psi(R'_{-j})| \Leftrightarrow \rho_j^\psi(R') \text{ first order stochastically dominates } \rho_i^\phi(R)
\end{equation*}
\end{proposition}
The proof is provided in Appendix~\ref{appendix:delta-rank}. It relies on an explicit characterisation of the probability law of $\rho_i^\phi(R)$, conditional on $|\mathcal{M}_{i|i}^\phi(R_{-i})|$.

Note that this proposition does not rule out $\phi=\psi$ or $i=j$. Hence, if we set $i=j$, they connect the freedom of choice (in a cardinal sense) that an agent enjoys under two mechanisms to the welfare they enjoy under the two mechanism. Alternatively, holding the mechanism constant, we can compare the situation of two agents in the same mechanism,  and find that if one agent enjoys more freedom than another, they will fare better in a stochastic dominance sense.

The result can be considered a cardinal analogue of lemma~\ref{lemma:indiv.menu.welfare}: both show that an agent has a larger menu (either in an inclusion or a cardinal sense) if and only if she is better off in terms of welfare (either in deterministic terms, or in the sense of a higher probability of obtaining an outcome ranked higher in her preferences list).

Can this result be generalised when others’ preferences are random too? Under the assumption of $R_{i}$ (resp. $R'_{j})$ uniformly random and independent from $R_{-i}$ (resp. $R'_{-j}$) -- the “impartial culture” assumption often invoked in voting theory --, Appendix~\ref{appendix:delta-rank} provides elements of an answer: first, the forward direction from the equivalence above remains true (corollary~\ref{cor:delta-rank}); second, the expected size of the menu is proportional to the probability of obtaining the preferred outcome, i.e. $ \mathbb{E}[|\mathcal{M}_{i|i}^\phi(R_{-i})|] = | A_i|\times \mathbb{P}[\rho_i^\phi(R) = 1]$ (corollary~\ref{cor:probability-preferred-outcome}).

To conclude this section, note that a possible use of the cardinal index is in cases where we do not know with certainty $R_{-i}$, the other players' preference profile, but only a probability distribution allowing us to compute the expected cardinal index.

\begin{example}
    We have mentioned above that in a binary voting game $\phi$, under a given profile of opposing preferences, an agent's menu includes the two possible outcomes if she is decisive, but only a singleton if she is not. Therefore, under a given probability distribution for $R_{-i}$:
    \begin{equation*}
        \mathbb{E}[|\mathcal{M}_{i|i}^\phi(R_{-i})|] = 1 + \mathbb{P}[i \text{ is decisive}]
    \end{equation*}
    But we know from \citet{Straffin1988} that the canonical voting power indices \citep{ShapleyShubik1954, Banzhaf1965} correspond to the probability of being decisive under certain assumptions on the joint probability distribution of other agents’ preferences. In the context of binary voting games, there is therefore a direct relationship between our cardinal measure of power and the canonical indices\footnote{Cf. Appendix C of \cite{basteck2026power}.}.
\end{example}

\section{Power over others}
\label{section:power-over-others}
Variations of the same tools allow to describe the influence that an agent has over others. Let us first modify definition~\ref{def:menus} to describe the menu of an agent for the outcome of a group of agents.

\begin{definition}[Menu for others] Consider $\phi$ is a strategy-proof mechanism, $N$ the set of agents, $i \in N$ an agent, $S \subseteq N$ a subset of agents, $\mathcal{R}_1, \mathcal{R}_2, ...$ the preference domain of the agents and $R_{-i} \in \mathcal{R}_{-i}$ a preference profile for agents distinct from i. Then we define the \textit{choice set} or \textit{menu of agent $i$ for group $S$} as:
    \begin{equation*}
        \mathcal{M}^\phi_{i|S}(R_{-i}) := \{ ([\phi(R_i,R_{-i})]_j)_{j \in S} | R_i \in \mathcal{R}_i\}
    \end{equation*}
\end{definition}

In other words, her menu for group $S$ is the set of tuples of outcomes relevant to agents in $S$ obtained when $i$'s preferences vary in her domain, the others' preferences being fixed. When concerned with the power of an individual over another, we will write $\mathcal{M}_{i|j}$ for $\mathcal{M}_{i|\{j\}}$.

We will interpret the comparison of these menus for others as comparisons of power: for example, if $\mathcal{M}_{i|S}^\phi(R_{-i}) \supset \mathcal{M}_{i|S}^\varphi(R_{-i})$, then we will say that under preference profile $R_{-i}$, mechanism $\phi$ provides agent $i$ with more power over group $S$ than mechanism $\varphi$. As previously with freedom, this is not an instrumental conception of power. Power is not formalised here as a lever for exerting pressure on others in order to increase one’s own well-being, but rather as the influence that an agent pursuing their own ends may have over others. This brings us closer, arguably, to the non-welfare based concern with domination developed by the philosophers cited in the introduction, such as \citet{Anderson1999} and \citet{Pettit2006}. 

As we did with freedom, we can extend comparability by measuring the cardinality of these sets when they are finite, defining the \textit{cardinal index of power of agent $i$ over group $S$} as $| \mathcal{M}^\phi_{i|S}(R_{-i}) |$.

Let us anticipate slightly on the next section by illustrating these two new definitions with the example of a serial dictatorship mechanism.

\begin{example*}
    Consider an assignment mechanism $\phi$ that matches three agents $1, 2, 3$ and three objects $a, b, c$. Each agent has strict preferences over the objects. The mechanism assigns agent $1$ their $R_1$-most-preferred object, then agent $2$ their $R_2$-most-preferred among the remaining ones, and finally agent $3$ obtains the remaining object. Suppose that agent 2 and 3 share the same preferences $R_2 = R_3:abc$.

    When agent 1 prefers object $a$, 2 obtains object $b$ and 3 obtains $c$. When 1 prefers $b$, 2 obtains $a$ and 3 obtains $c$. When 1 prefers $c$, 2 obtains $a$ and 3 obtains $b$. So the menu of 1 for $\{2,3\}$ together is $\mathcal{M}_{1|-1}^\phi(R_{-1}) = \{ (b,c), (a,c), (a,b) \}$, of cardinality $|\mathcal{M}_{1|-1}^\phi(R_{-1})| = 3$, while the menu of 1 for the individual agent 2 is $\mathcal{M}_{1|2}^\phi(R_{-1}) = \{ a,b\}$, of cardinality $2$.
\end{example*}

Our next result shows that under some conditions, there is a connection between the freedom of an agent and her power over all others, in a cardinal sense. To state the result, we first need to introduce the notions of non-bossiness and non-autarky.

\begin{definition}[Non-bossiness]
    Under a mechanism $\phi$ and given others' preferences $R_{-i}$, agent $i$ with preference domain $\mathcal{R}_i$ is \textit{non-bossy} if she cannot change another agent's outcome without changing her own, i.e., if:
    \begin{equation*}
        \forall R_i, R'_i \in \mathcal{R}_i: [\phi(R_i, R_{-i})]_{i} = [\phi(R'_i, R_{-i})]_{i} \Rightarrow [\phi(R_i, R_{-i})]_{-i} = [\phi(R'_i, R_{-i})]_{-i}.
    \end{equation*}
    
 If every agent, at all profiles of others' preferences, is non-bossy, then so is the mechanism.
\end{definition}

In the following propositions, we consider the mechanism $\phi$ and the profile of opposing preferences $R_{-i}$ as given and drop the related notations, i.e., write $\mathcal{M}_{i,\cdot}$ for $\mathcal{M}^\phi_{i,\cdot}(R_{-i})$.
\begin{lemma}
\label{lemma:non-bossy-implication}
    If agent $i$ is non-bossy, then $|\mathcal{M}_{i|-i}| \leq |\mathcal{M}_{i|i}|$.
\end{lemma}
\begin{proof}
   Denote $\mathcal{F}_i$ (resp. $\mathcal{F}_{-i}$) the partition of $\mathcal{R}_i$ by fibers of $[\phi(.,R_{-i})]_i$ (resp. $[\phi(.,R_{-i})]_{-i}$). In other words, distinct preferences from $\mathcal{R}_i$ are regrouped in the same part of $\mathcal{F}_i$ if they give rise to the same outcome for $i$.
   
   Note that the non-bossiness of $i$ is equivalent to the claim that when two elements are in the same part of $\mathcal{F}_i$, they are also in the same part of $\mathcal{F}_{-i}$, i.e., $\mathcal{F}_{-i}$ is coarser than $\mathcal{F}_i$. So the number of parts in $\mathcal{F}_{-i}$ is lower than in $\mathcal{F}_i$. But the number of parts in $\mathcal{F}_{-i}$ (resp. $\mathcal{F}_i$) is $|\mathcal{M}_{i|-i}|$ ($|\mathcal{M}_{i|i}|$).
\end{proof}


This result offers a new interpretation of non-bossiness:\footnote{See \citet{Thomson2016} for a survey of existing justifications of this technical property.} being bossy doesn't mean to have power over others, but to have power over others in excess to one's freedom or power over oneself. Now let us introduce a somewhat symmetric notion.

\begin{definition}[Non-autarky]
    Under mechanism $\phi$ and facing opposing preferences $R_{-i}$, agent $i$ is \textit{non-autarkic} if she cannot change her own outcome without changing that of at least one other agent, i.e., if
    \begin{equation*}
    \forall R_i, R'_i \in \mathcal{R}_i: [\phi(R_i,R_{-i})]_{-i} = [\phi(R'_i,R_{-i})]_{-i} \Rightarrow [\phi(R_i,R_{-i})]_{i} = [\phi(R'_i,R_{-i})]_{i}
    \end{equation*}
\end{definition}

The mechanism is non-autarkic if every of its agents, given any preference profile for the others, is. A proof similar to that of the previous lemma shows that:

\begin{lemma}
\label{lemma:non-autarkic-implication}
    If agent $i$ is non-autarkic, then $|\mathcal{M}^\phi_{i|-i}| \geq |\mathcal{M}^\phi_{i|i}|$.
\end{lemma}
Combining these two lemmas provides the following result.
\begin{proposition} \label{prop:non-autarkic-and-non-bossy}
    Consider a mechanism $\phi$ with a finite outcome set, and an agent $i$ with preference domain $\mathcal{R}_i$. If two of the following three claims are satisfied, the third is satisfied as well: 
    \begin{itemize}
        \item[(i)] $i$ is non-autarkic.
        \item[(ii)] $i$ is non-bossy.
        \item[(iii)] $\forall R_{-i} \in \mathcal{R}_{-i}, \; |\mathcal{M}_{i|-i}^\phi(R_{-i})| = |\mathcal{M}_{i|i}^\phi(R_{-i})|$.
    \end{itemize}
\end{proposition}
\begin{proof}
    \textit{(i) and (ii) $\Rightarrow$ (iii)}: This follows directly from lemmas~\ref{lemma:non-bossy-implication} and \ref{lemma:non-autarkic-implication}, $|\mathcal{M}_{i|-i}| = |\mathcal{M}_{i|i}|$.

      \textit{(i) and (iii) $\Rightarrow$ (ii)}: Consider the partitions $\mathcal{F}_i$ and $\mathcal{F}_{-i}$ as defined in the proof of lemma~\ref{lemma:non-bossy-implication}. Because $|\mathcal{M}_{i|i}| = |\mathcal{M}_{i|-i}|$, both partitions have the same number of parts. Because of non-autarky, $\mathcal{F}_{i}$ is coarser than $\mathcal{F}_{-i}$. So both partitions are the same, so $\mathcal{F}_{-i}$ is also coarser than $\mathcal{F}_{i}$, i.e., $i$ is non-bossy.

     By a similar reasoning,  \textit{(ii) and (ii) $\Rightarrow$ (i)}.
\end{proof}

Non formally, in mechanisms where an agent cannot affect others without affecting her own outcome and \textit{vice versa}, her power over others (measured by the tuples of others' outcome) coincides with her freedom (her power over her own outcome).

These twin conditions are satisfied in voting games where the same outcome is relevant to all agents, so that the only way to affect either oneself or others is to change the outcome of the vote, which makes the mechanism both non-bossy and non-autarkic. It is also the case for many common assignment mechanisms that assign one object to each agent, provided there are as many objects as agents. It is to these mechanisms that we turn now.

\section{Assignment problems}
\label{section:assignment}
In this section, we analyse (object) assignment rules in terms of the freedom of choice and power over others that they award to agents. We start by showing some general implications of efficiency for the structure of agents' menus: at some preference profiles, some agent must have unconstrained freedom of choice while another agent must have no choice. Further, we show that hierarchical exchange rules \citep{Papai2000} are indeed hierarchical in a hitherto unobserved sense: whenever an agent $i$ has any power over another agent $j$, she also has power over any other agent that $j$ has power over. Moreover, this transitive nature of power characterises hierarchical exchange rules among all group-strategy-proof and efficient rules. 

Then we study separately the properties of two canonical rules: (Bipolar) Serial Dictatorship (SD) and Top Trading Cycles (TTC). Among group-strategy-proof, and efficient rules, we show that (bipolar) SD are characterised by a minimal bilateral power property -- the menu of $i$ for $j$ is at most of size 2. In contrast, under TTC, there exist preference profiles where some agent $i$ can assign any of all $n$ objects to another agent $j$ by varying her own preferences. While  `dictatorial' in that aspect, TTC also offers the possibility of complete freedom -- it awards each agent unconstrained freedom of choice at at least some opposing preference profile. Moreover, among group-strategy-proof and efficient rules of four or more agents, TTC is characterised by that property. 

Together, these results yield a form of impossibility theorem: a group-strategy-proof and efficient rule can ensure that no participant exerts more than minimal power on another, or can provide each participant the possibility of unconstrained freedom at some profile of the others' preferences, but cannot ensure both when there are 4 or more agents.

\medskip

\subsection{Model}
As before, we consider a set of $n$ agents $i\in N$. The task is to assign to each of them one of $n$ objects $x\in O$ so that the set of outcomes corresponds to the set of all bijections $\mu:N\rightarrow O: i\mapsto \mu_i$. We refer to such outcomes as assignments and denote the set of all possible outcomes as $\mathcal{A}$. Each agent is assumed to have strict preferences over the object that they receive and are otherwise indifferent between assignments.  Hence, for a given assignment $\mu$ we can identify agents' $i$-relevant outcome $[\mu]_i$ with the object they receive, $\mu_i$. Let $\mathcal{R}=\times_{i\in N} \mathcal{R}_i$ denote the domain of all such preference profiles.\footnote{Note that the domain is is top-rich, i.e., contains for each agent $i$ and each $i$-relevant outcome $[x]_i$ a preference order $R_i$ where $[x]_i$ is ranked at the top.} Finally, a social choice rule in this setting, mapping preference profiles to assignments, is referred to as an assignment rule. The object assigned to $i$ under assignment rule $\varphi:\mathcal{R} \rightarrow \mathcal{A}: R\mapsto\mu$ is denoted as $\varphi_i(R)$.



An assignment rule is efficient if $\varphi(R)$ is efficient for every $R\in \mathcal{R}$; it is group-strategy-proof if it is not manipulable by a group, i.e., if there is no $R\in \mathcal{R}$, $G\subseteq N$, and $R'_G \in \mathcal{R}_G$ s.t. $\varphi_i(R'_G,R_{-G}) R_i \varphi_i(R)$ for all $i\in G$ and $\varphi_{j}(R'_G,R_{-G}) P_{j} \varphi_{j}(R)$ for some $j\in G$. On the domain of assignment problems, group-strategy-proofness is equivalent to strategy-proofness and non-bossiness \citep{Papai2000}.

\cite{pycia2017incentive} provide a characterisation of all efficient and group-strategy-proof assignment rules in terms of algorithmic procedures where some objects are assigned to some agents (taking into account agents' preferences) and thus removed, some remaining objects are assigned to some remaining agents (taking into account the previous assignment of objects as well as the preferences of so far unmatched agents) and thus removed, etc. To summarise their result we need to introduce some of their additional notation.

First, we define a \emph{submatching} $\sigma$ as a bijection between a subset $N_\sigma\subseteq N$ and a subset of $O_\sigma \subseteq O$. We say that a submatching is a proper submatching if $N_\sigma \subsetneq N$ and denote the set of all proper submatchings as $\mathcal{S}$. The set of unmatched agents and objects (under $\sigma$) are denoted $\overline{N}_\sigma$ and $\overline{O}_\sigma$. 
We say that $\sigma$ is a submatching of $\sigma'$ if $N_\sigma\subseteq N_{\sigma'}$ and $\sigma_i=\sigma'_i$ for all $i\in N_\sigma$. Interpreting submatchings as subsets in $N\times O$, we also write $\sigma\subseteq \sigma'$. The \textit{empty submatching} $\sigma$ is such that $N_\sigma = O_\sigma = \emptyset$; for simplicity we denote it as $\sigma = \emptyset$.

Next, a \emph{structure of control rights}  $c$ maps each proper submatching $\sigma \in \mathcal{S}$ and each unassigned object $x\in \overline{O}_\sigma$ to an unassigned agent who controls it, either as an \emph{owner} ($o$) or as a \emph{broker} ($b$), i.e., $c_\sigma(x)\in \overline{N}_\sigma\times \{o,b\}$, such that (i) no agent is both a broker for one and an owner for another object (while owners may own multiple objects); (ii) if there are  $|\overline{N}_\sigma|\neq 3$ unassigned agents, then there can be at most one broker, (iii) if $|\overline{N}_\sigma|=1$, then the last remaining agent is an owner of the last remaining object, and (iv) if  $|\overline{N}_\sigma|=3$ there can be no, one or three brokers, each broker brokering a different object.


For $|N|=|O|>3$, a trading cycle (TC) mechanism associated with an control rights structure $c$, is then defined by the following algorithm:

\begin{itemize}
    \item Start with the empty submatching $\sigma=\emptyset$. Let all owners point to their most preferred object and, if there exists a broker, let her point to her most preferred object among those that she does not broker. Each object points to the agent who controls it. Each agent in a cycle is assigned the object that they point to, giving rise to a submatching $\sigma'$.\footnote{There must be at least one cycle. If there is more than one, one could equivalently clear cycles sequentially -- it can be shown that the final allocation does not depend on the order in which they are selected (cf. footnote 20 in \cite{pycia2017incentive}).} 
    \item The submatching $\sigma'$, together with the control rights structure $c$ gives rise to new ownership and brokerage rights. If there is at most one broker, proceed as before.
    \item Repeat until (i) all objects are assigned (with no submatching in the sequence giving rise to three brokers) or, (ii) at the point where 3 unassigned agents and objects remain, and each agent brokers one of the objects -- in which case the assignment of the remaining objects is given by an avoidance matching that chooses an efficient assignment of the remaining objects that minimises the number of objects assigned to their respective broker.\footnote{If there are multiple efficient matchings that assign an equal number of objects to their respective brokers, additional case distinctions apply, see \cite{Bade2020} or the online appendix to \cite{pycia2017incentive}.}
    \end{itemize}

For this algorithm to give rise to a an assignment rule that is not only efficient but also group-strategy-proof, \cite{pycia2017incentive} identify a number of necessary conditions on the control rights structure. Moreover, and importantly, they show that any efficient and group strategy proof assignment rule can be represented as a TC-mechanism satisfying their conditions. We do not restate all these conditions here, but among them, the following is crucial for our results:  for any $\sigma\subseteq \sigma'$, any unassigned agent in $\sigma'$, $i\in \overline{N}_{\sigma'}$, and any unassigned object in $\sigma$, $x\in \overline{O}_{\sigma}$ we have that $c_{\sigma}(x)=(i,o) \implies c_{\sigma'}(x)=(i,o)$, i.e., ownership persists until an agent is matched eventually.

TC-mechanisms involving brokers may allow pairs of agents to jointly misrepresent their preferences and to thus be assigned different objects such that, if they were to exchange their assigned objects, both agents would weakly better off, at least one of them strictly. Ruling out such situations by requiring the assignment rule to be \emph{reallocation-proof},\footnote{Formally, an assignment rule $\varphi$ violates reallocation proofness iff there exist $i,j\in~N$, $R\in \mathcal{R}$, and $R_i'\neq R_i$, $R_j'\neq R_j$, such that 
$\varphi_j(R_i',R_j',R_{-\{i,j\}})R_i\varphi_i(R)$, 
$\varphi_i(R_i',R_j',R_{-\{i,j\}})P_j\varphi_j(R)$, 
and $\varphi_h(R)=\varphi_h(R_h,R_{-h})\neq
\varphi_h(R_i',R_j',R_{-\{i,j\}})$ 
for $h\in\{i,j\}$.} 
\cite{Papai2000} characterises \emph{hierarchical exchange rules} as the set of all efficient, group-strategy-proof, and reallocation-proof assignment rules. They correspond to the class of assignment rules that can be implemented by a TC-mechanism without brokers, i.e., where at each possible submatching, any agent controlling an object owns that object. Thus, we may think of hierarchical exchange rules, as TC-mechanisms where the control rights structure is a mapping $c_\sigma: \overline{O}_\sigma\rightarrow \overline{N}_{\sigma}$ with persistent ownership. 

Note that a given assignment rule may be represented by different, equivalent TC-mechanisms -- or more precisely by different control right structures.\footnote{For example, to represent a serial dictatorship where agent $1$ chooses first, we need to make $1$ the owner of all objects at the empty submatching, but can choose control rights at a submatching where only agents other than $1$ are matched in any arbitrary way since these will never be reached in the algorithmic procedure. Even for submatchings that can be reached, different assignments of control rights may be equivalent -- for example, if only two unmatched agents are left, it does not matter whether one of them is a broker (while the other agent owns the remaining object) or whether the other agent owns both objects.} In particular, any hierarchical exchange rule admits a representation as a TC-mechanism without brokers, even if there exist equivalent TC-mechanisms involving brokers.

\subsection{Results}
Our first observation concerns efficient (and strategy-proof) assignment rules. Necessarily, they will treat agents very differently in terms of their effective freedom of choice, at least at some preference profiles. 

\begin{proposition}
\label{prop:assignment-extreme-situations}
Let $\varphi: \mathcal{R}\rightarrow A$ be an efficient and strategy-proof assignment rule. Then there exists a preference profile $R\in\mathcal{R}$  and an ordering of agents $i_1,\dots, i_n$ such that for any $1\leq k\leq n$, the size of $i_k$'s menu is $k$, i.e., $|\mathcal{M}^\varphi_{i_k|i_k}(R_{-i_k})|=k$. 
\end{proposition}
\begin{proof}
    Consider a preference profile $R\in\mathcal{R}$ where all agents have identical preferences, i.e., $R_i:o_n,o_{n-1},\dots, o_{1}$ for all $i\in N$. Take any $k$, $1\leq k\leq n$, and consider the agent who is assigned object $o_k$, i.e., $i$ such that $\varphi_i(R)=o_k$. Denote that agent as $i_k$. We claim that the menu of $i_k$ consist of all objects ranked weakly below $o_k$, i.e., $\mathcal{M}_{i_k|i_k}(R_{-i_k})=\{o_k,\dots,o_1\}$. 
    
    Towards a contradiction, suppose there is some higher ranked object $o_l$, $l>k$, included in $i_k$'s menu. But then there exists some $R'_{i_k}$ such that $\varphi_{i_k}(R'_{i_k},R_{-{i_k}})=o_l$, violating strategy proofness for $i_k$ (at $R$).     
    
    It remains to show that $i_k$ can not only be assigned object $o_k$ but any object $o_l$, $l < k$, for some $(R'_{i_k},R_{-i_k})$.
    Take any $o_l$, $l< k$, and let $R'_{i_k}$ be such that $o_l$ is ranked first, while the ranking of other objects remains unchanged, $R'_{i_k}:o_l,o_n, o_{n-1},\dots,o_1$. By strategy-proofness, we have  $\varphi_{i_k}(R'_{i_k},R_{-{i_k}})\in \{o_k,o_l\}$. If $\varphi_{i_k}(R'_{i_k},R_{-{i_k}})=o_k$, there is some other agent $j$ for whom $\varphi_{j}(R'_{i_k},R_{-{i_k}})=o_l$. But then there exists a  Pareto improving trade between $i_k$ and $j$, contradicting efficiency of $\varphi$. Thus we conclude that $\varphi_{i_k}(R'_{i_k},R_{-{i_k}})=o_l$.
\end{proof}




In other words, under an efficient and strategy-proof assignment rule, agents will, at times, enjoy very different degrees of freedom of choice. For efficient and group-strategy-proof assignment rules, proposition~\ref{prop:non-autarkic-and-non-bossy} implies that agents' power over others will be equally unevenly distributed.\footnote{Note that group-strategy-proofness implies non-bossiness while efficiency (and feasibility) imply non-autarky.}

Since such inequality is unavoidable, one may ask whether on average some rules  provide agents with more freedom than other rules. For this, our next proposition compares efficient and group-strategy-proof assignment rules under the `impartial culture' benchmark.

\begin{proposition}\label{prop:same.expected.menusize}
Let $\phi: \mathcal{R}\rightarrow A$ be an efficient and group-strategy-proof assignment rule. Suppose agents' preferences are drawn independently, with each $R_i$ uniformly distributed on $\mathcal{R}_i$. Then the average expected cardinal index of freedom is given by:
    \begin{equation*}
        \mathbb{E}\left[ \frac{1}{n} \sum_{i =1}^n |\mathcal{M}_{i|i}^\phi(R_{-i})|\right] = \frac{n+1}{2}
    \end{equation*}
\end{proposition}

\begin{proof}
When $R_i$ is drawn uniformly among strict orderings of the elements of $O$, the most preferred object $R_{i,1}$ is drawn uniformly in $O$ so that corollary~\ref{cor:probability-preferred-outcome} applies -- for given $R_{-i}$, $i$'s freedom index coincides with her probability of getting her most preferred object: $|\mathcal{M}_{i|i}^\phi(R_{-i})| = n\mathbb{P}[\rho_i^\phi(R) = 1|R_{-i}]$ (recall from definition~\ref{def:rank} that $\rho_i^\phi(R)$ is the rank in $R_i$ of $i$'s assignment). Hence
\begin{align*}
    \mathbb{E}\left[ \frac{1}{n} \sum_{i =1}^n |\mathcal{M}_{i|i}^\phi(R_{-i})|\right]
    & = \sum_{i=1}^n \mathbb{E}\left[\mathbb{E}\left[ \frac{1}{n} |\mathcal{M}_{i|i}^\phi(R_{-i})||R_{-i}\right]\right] & \text{by the law of total expectation}\\
    & =\sum_{i=1}^n \mathbb{E}\left[  \mathbb{P}[\rho_i^\phi(R) = 1|R_{-i}]\right] & \text{by corollary~\ref{cor:probability-preferred-outcome} (Appendix~\ref{appendix:delta-rank})} \\
    & = \sum_{i=1}^n \mathbb{P}[\rho_i^\phi(R) = 1] & \text{by the law of total expectation} \\
    & = n\mathbb{P}[\rho_X^\phi(R) = 1] &
\end{align*}

where $X$ is an agent drawn uniformly in $N$, independently from all random variables introduced before. This allows us to use \citet{Bade2020}, who shows that given $n$ agents and objects, $R$ a given preference profile, and $X$ a uniformly random agent, the probability distribution of $\phi_X(.)]$, the function that maps a preference profile into $X$'s assignment, is the same for any assignment mechanism $\phi$ that is efficient, strategy-proof and non-bossy. But the event $\{ \rho_X^\phi(R) = 1\}$ can be rewritten as $\{ \phi_X(R) = R_{X,1}\}$, so is a function of the function $\phi_X(.)$, so has also the same probability for any $\phi$, and in particular is the same as under a serial dictatorship mechanism $\psi$. So $\mathbb{E}\left[ \frac{1}{n} \sum_{i =1}^n |\mathcal{M}_{i|i}^\phi(R_{-i})|\right] = \mathbb{E}\left[ \frac{1}{n} \sum_{i =1}^n |\mathcal{M}_{i|i}^\psi(R_{-i})|\right]$.

Anticipating slightly on the formal definition of serial dictatorship, the following intuitive claim can already be made: under SD, the menu of the first agent is of size $n$ (she can choose any object), of the second agent of size $n-1$ (she can choose any object but the one picked by agent 1), ..., of the last agent of size $1$. But SD is efficient and group-strategy-proof, so:
\begin{equation*}
   \mathbb{E}\left[ \frac{1}{n} \sum_{i =1}^n |\mathcal{M}_{i|i}^\psi(R_{-i})|\right] = \frac{1}{n} (n + ... + 1) = \frac{n+1}{2}
\end{equation*}
\end{proof}

The next result establishes a connection between the influence that two agents may have on their own relevant outcome, comparing their freedom of choice, and whether one agent may have power over another.

\begin{lemma}
\label{lemma:menu-inclusion-implies-power}
    Consider a strategy-proof assignment rule $\phi$ for a set $N$ of agents, a preference profile $R$, and two agents $i,j \in N$. If $i$'s menu is a superset of $j$'s, then $i$ has power over $j$: \[\mathcal{M}_{i|i}^\phi(R_{-i}) \supseteq \mathcal{M}_{j|j}^\phi(R_{-j}) \implies |\mathcal{M}_{i|j}^\phi(R_{-i})| > 1.\]
\end{lemma}
\begin{proof}
    Given $R$, suppose that $j$ is assigned $\phi_j(R)=x\in \mathcal{M}_{j|j}^\phi(R_{-j})\subseteq \mathcal{M}_{i|i}^\phi(R_{-i})$. Since $x\in \mathcal{M}_{i|i}^\phi(R_{-i})$, there exists $R_i'$ s.t. $x=\phi_i(R_i',R_{-i})\neq\phi_j(R_i',R_{-i})$. But then $\mathcal{M}_{i|j}^\phi(R_{-i})=\{\phi_j(\tilde{R}_i,R_{-i})|\tilde{R}_i\in \mathcal{R}_i\}\supseteq \{\phi_j(R),\phi_j(R'_{i},R_{-i})\}$, hence $|\mathcal{M}_{i|j}^\phi(R_{-i})| > 1$.
\end{proof}

One may consider the requirement that one agent's menu is a superset of another agent's menu to be exceptionally strong. Is it ever satisfied beyond explicitly hierarchical rules such as Serial Dictatorship or \citet{PiccioneRubinstein2007}'s “jungle” model, its counterpart for divisible goods? The next section answers in the affirmative: for many assignment rules it is satisfied whenever one agent is able to influence the outcome of the other agent in any way.

\subsection*{Hierarchical exchange rules}
As already mentioned, applying proposition~\ref{prop:non-autarkic-and-non-bossy} to group-strategy-proof and efficient mechanisms ensures that in such a mechanism, each agent has as much freedom, in the cardinal sense, as power over all others together. If we focus on hierarchical exchange mechanisms, i.e., if we impose reallocation-proofness alongside efficiency and group-strategy-proofness, we uncover a direct link between freedom and power (this time power over another individual agent).

\begin{proposition}
\label{prop:hierarchical-exchange-menu}
    Consider a hierarchical exchange rule $\phi$ for a set $N$ of agents, a preference profile $R$, and two agents $i,j \in N$. The following statements are equivalent:
    \begin{enumerate}
        \item[(i)] $i$'s menu is a superset of $j$'s: $\mathcal{M}_{i|i}^\phi(R_{-i}) \supseteq \mathcal{M}_{j|j}^\phi(R_{-j})$,
        \item[(ii)] $i$ has power over $j$: $|\mathcal{M}_{i|j}^\phi(R_{-i})| > 1$,
    \end{enumerate}
\end{proposition}
\begin{proof}

    $(i)\implies(ii)$: see Lemma~\ref{lemma:menu-inclusion-implies-power} above.

    $(ii)\implies(i)$: Since $\phi$ is a hierarchical exchange rule, it can be represented as a TC-mechanism without brokers. Moreover,
    recall that the final matching under a TC-mechanism does not depend on the order in which trading cycles are resolved.\footnote{See footnote 20 in \cite{pycia2017incentive}.} 
    Hence, we can successively resolve all trading cycles that do not involve $i$ or $j$ (and which therefore do not depend on 
    $R_i$ or $R_{j}$ but only on $R_{N\backslash\{i,j\}}$) to eventually arrive at a submatching $\sigma$.
    Now, given $\sigma$ and the associated assignment of ownership,  consider the subgraph where all unmatched agents other than $i$ and $j$ point to the owner of their most preferred object. By construction of $\sigma$, there are no further cycles, so that the graph consist of two separate components, namely a  tree with root $i$ and a tree with root $j$ (possibly with $i$ or $j$ as isolated vertex, i.e., as a degenerate tree). As $\sigma$ before, this subgraph is independent of $R_i$ and $R_j$, but depends only on $R_{N\backslash\{i,j\}}$.  Now if $R_j$, is such that $j$'s most preferred object is owned by an agent in the component that forms a tree with root $j$ (possibly $j$ herself), then regardless of $\tilde{R}_i$, there will be a trading cycle involving $j$ where $j$ receives her most preferred among all remaining objects in $\overline{O}_\sigma$. Hence, $i$ would have no power over $j$. Instead we know that $R_j$ must be such that $j$'s most preferred object is owned by an agent on the component that forms a tree with root in $i$. Now, for any $\tilde{R}_i$, if we consider the owner(s) to which $i$ and $j$ point, there is a cycle that can be cleared and in which $i$ receives their most preferred among the remaining objects. Thus, $\mathcal{M}_{i|i}(R_{-i})=\overline{O}_\sigma$. Moreover, regardless of $R_i$, we know that $j$ can receive at most any object in $\overline{O}_\sigma$ by varying their own reported preferences. Thus, $\mathcal{M}_{i|i}^\phi(R_{-i}) \supseteq \mathcal{M}_{j|j}^\phi(R_{-j})$.
\end{proof}

Proposition \ref{prop:hierarchical-exchange-menu} yields the following corollary (which follows from the fact that set-inclusion is a transitive relation):

\begin{corollary}\label{cor:power.transitive}
    Consider a hierarchical exchange rule $\phi$ for a set $N$ of agents, a preference profile $R$, and three agents $i,j,k \in N$. If $i$ has power over $j$, and $j$ has power over $k$, then $i$ has power over $k$:
    \[ |\mathcal{M}_{i|j}^\phi(R_{-i})| > 1 \text{ and } |\mathcal{M}_{j|k}^\phi(R_{-j})| > 1 \implies |\mathcal{M}_{i|k}^\phi(R_{-i})| > 1.\]
\end{corollary}

The transitivity of power relations revealed in Corollary \ref{cor:power.transitive} is specific to hierarchical exchange rules but breaks down for general TC-mechanisms, that allow for brokers:

\begin{example}\label{ex.nontrans}
    Consider a TC-mechanism with at least three agents and objects. Let agent $1$ be the broker for $a$, agent $2$ the owner of $b$, and agent $3$ the owner of $c$. If preferences $R$ are such that $1$ considers $b$ most preferred, $2$ prefers $c$ the most and $3$ prefers $a$, then all three agents are assigned their most preferred object. If instead $1$ would report $c$ as most preferred, they would be assigned $c$ and $2$ would be left with the (less-preferred) object $b$. Thus, $1$ has power over $2$. Similarly, $2$ has power over $3$, as they could change the assignment of $3$ by instead reporting $a$ as most preferred (thus trade their object with $1$). However, $1$ has no power over $3$: regardless of whether they report $b$ or $c$ as most preferred, the resulting trading cycle leaves $3$ with object $a$ (recall that as a broker for $a$, agent $1$ will be forced to trade it, given $R_{-1}$).
\end{example}

Together, example~\ref{ex.nontrans} and corollary~\ref{cor:power.transitive} characterise hierarchical exchange rules as the only efficient and group-strategy-proof assignment rules where power over another forms a transitive relation. This is an alternative to \citet{Papai2000}'s original characterisation, which relies on group-strategy-proofness, efficiency, and reallocation-proofness. Our characterisation may offer a formal rationale for the name of this class of mechanisms, by identifying “hierarchy” with power transitivity.

\begin{corollary}\label{cor:hierarchical}
    Let $\phi: \mathcal{R}\rightarrow A$ be an efficient and group-strategy-proof assignment rule. Then $\phi$ is a hierarchical exchange rule if and only if for all agents $i,j,k\in~N$ and at all preference profiles $R\in\mathcal{R}$ where $i$ has power over $j$, and $j$ has power over $k$, we find that $i$ has power over $k$.
\end{corollary}

\begin{proof}
    In light of corollary \ref{cor:power.transitive} we only need to show the other direction. So suppose that $\phi$ is not a hierarchical exchange rule, thus cannot be represented by TC-mechanism without brokers. If it is represented by a TC-mechanism where at some submatching (that can be reached given some preference profile) we have a broker, say $i$, facing two owners, say $j$ and $k$, a construction analogous to example \ref{ex.nontrans} establishes that $i$ has no power over $k$. Similarly, if $i$ faces additional owners beyond $j$ and $k$ -- we may then consider a preference preference profile where the additonal owners consider one of their own objects to be most-preferred, while objects owned by them are least-preferred by $i$, $j$, and $k$.
    
    If $\phi$ can be represented by a TC-mechanism where at some submatching a broker faces exactly one owner (but never more), consider a minimal submatching $\sigma$ for which this is the case (i.e., such that at all submatchings $\sigma'\subsetneq \sigma$, there are no brokers). Then, the broker, say $i$, is forced to point to the owner, say  $j$, regardless of their own preferences so that the owner will be assigned their most-preferred object regardless of their preferences. If instead we change control rights at $\sigma$ and make $j$ the owner of all unassigned objects (including the object previously brokered by $i$), while keeping the control rights at all other submatchings unchanged, the new TC-mechanism is equivalent to the previous one. By iteratively removing brokers at any minimal submatching where they exist, we eventually arrive at an equivalent TC-mechanism without brokers -- contradicting the initial assumption that $\phi$ is not a hierarchical exchange rule.
\end{proof}


\subsection*{Top trading cycle}

As we mentioned above, all efficient and group-strategy-proof assignment rules will, at some profiles, allow only one of the agents to choose from a complete menu (and have maximal power over others) -- while another agent has no influence on their own (or any other) relevant outcome at all (Proposition~\ref{prop:assignment-extreme-situations} combined with Proposition~\ref{prop:non-autarkic-and-non-bossy}).

Such extreme inequity cannot be avoided, but it can take different forms. In the remainder of the paper, we contrast the patterns of power and freedom in two important subclasses of hierarchical exchange rules: Top Trading Cycles (TTC) on the one hand, and bipolar Serial Dictatorship (SD)—an extension of classical serial dictatorship—on the other.
TTC is characterised, among TC rules, by universal possibility of complete freedom: every agent can, at some opposing preference profile, enjoy a position of unrestricted freedom. In the same class of rules, bipolar SD, in contrast, is characterised  by minimal bilateral power: the size of an agent’s menu for another agent never exceeds two.
In other words, TTC grants each individual the potential for full freedom but allows situations in which one agent may exert full power over another. Bipolar SD, by contrast, ensures that no agent can determine another’s outcome beyond what is necessary, yet it fully predetermines the hierarchy of who enjoys more and who enjoys less freedom.

\begin{definition}[Top Trading Cycle]
For $n$ agents and objects, an assignment rule is a Top Trading Cycle rule (TTC) iff it corresponds to a TC-mechanism where each agent is the initial owner of one of the objects.\footnote{The algorithm was proven to be strategy-proof by \citet{Roth1982}.} We refer to the object that an agent $i\in N$ owns initially as their \textit{endowment} and denote it as $\omega_i\in O$. 

\end{definition}

Now, because TTC is efficient and strategy-proof, we know from Proposition~\ref{prop:assignment-extreme-situations} that it is inevitable that there is a preference profile where one agent enjoys no freedom, while another enjoys complete freedom. Our Theorem~\ref{thm:ttc-char} below shows that TTC rules uniquely stand out among efficient and group strategy proof assignment rules in that they allow every agent to be in the fortunate latter position for at least some preference profiles\footnote{Appendix D of \cite{basteck2026power} connects that characterisation with the germane result of \citet{LongVelez2021} on balancedness.}.

But given that TTC is non-autarkic and non-bossy, we also know from Proposition~\ref{prop:non-autarkic-and-non-bossy} that the inequality of freedom just mentioned must also be reflected into relations of power over others. Actually, and perhaps surprisingly, our characterisation theorem also shows that the bilateral power of $i$ over $j$ can be maximally large under TTC: given two agents $i$ and $j$, there is a preference profile such that $i$ has full power over $j$.

\begin{theorem}\label{thm:ttc-char}
    Consider $\phi$ an efficient and group strategy-proof assignment rule with at least four objects (and agents). Then the following three claims are equivalent:
    \begin{itemize}
        \item[(i)] $\phi$ is a TTC-rule
        \item[(ii)] for each agent $i$, there is a preference profile $R$ such that $\mathcal{M}_{i|i}^\phi(R_{-i}) = O$
        \item[(iii)] for any two agents $i$ and $j$, there is a preference profile $R_{-i}$ such that $\mathcal{M}_{i|j}^\phi(R_{-i}) = O$
    \end{itemize}
\end{theorem}

\begin{proof}
    \textit{(i) $\Rightarrow$ (ii)} If $\phi$ is a TTC rule, consider a preference profile $R$ where there is a single top trading cycle that includes all agents -- say such that agents for all $j<n$, agent $j$ considers the object owned by $k=j+1$ most-preferred while agent $j=n$ considers $1$'s object best. Then for $i\in N$, every object is in her menu: by pointing to any agent $j$ (and their owned), $i$ creates a cycle including herself where she is assigned the object she pointed to.
    
    \textit{(ii) $\Rightarrow$ (i)} Suppose $\phi$ is not a TTC mechanism. Then, by \cite{pycia2017incentive}, $\phi$ must be a TC mechanism where initially some agent $i$ does not own any object.
    
    If $i$ controls no object (neither as owner nor as broker), then for any possible opposing preference profile $R_{-i}$ there will be a cycle among owner(s) and possibly a broker in which objects are removed in the first step -- none of these objects are then in the menu of $i$, $\mathcal{M}_{i|i}(R_{-i})$.
    
    If instead $i$ is the broker of some object $x$, i.e., if $c_\emptyset(x)=(i,b)$ consider the subgraph of owners pointing to their most preferred object (and of objects pointing to their owner). If there is a cycle among owners and their objects, then the involved objects are again removed in the first step of the algorithm and none of the objects are in $\mathcal{M}_{i|i}(R_{-i})$. If on the other hand, there is no cycle among owners and their objects, the subgraph involving owners and their controlled objects consists only of chains that originate at some owner and that end at in an edge where an owner points to $i$. Regardless of the object that $i$ points to, there will then be a cycle involving $i$ which implies that $i$ trades away there brokered object, i.e., $x\notin \mathcal{M}_{i|i}(R_{-i})$.
    Hence, we have $\mathcal{M}_{i|i}(R_{-i})\neq O$ for all $R_{-i}$.

    \textit{(i) $\Rightarrow$ (iii)}     If $\phi$ is a TTC rule, without loss of generality, we examine the power of 1 over $n$. Let us first prove it when $n$ is odd. Consider a preference profile where $R_n: \omega_1 \omega_2 ... \omega_n$ and the other agents (2, 3, ..., $n-1$) are organised by pairs, each having $\omega_1$ as most preferred and the endowment of the other member of the pair as second most preferred, i.e. $R_2 : \omega_1 \omega_3 \omega_n ...$, $R_3 : \omega_1 \omega_2 \omega_n ...$, etc. Then when 1 prefers her own endowment $R_1$, agents $2, 3, ..., n-1$ will trade two by two and agent $n$ will keep $\omega_n$. When 1 prefers $\omega_i$, with $i \in \{ 2, ..., n-1\}$, $1$ and $i$ will trade together. All the other pairs trade within themselves. The agent $j$ who was paired with $i$ trades with $n$, who gets object $\omega_j$. When 1 prefers $\omega_n$, 1 and $n$ trade and $n$ gets $\omega_1$. So depending on 1's preferences, $n$ can be assigned any object.

    When $n=3$, the property is verified with $R_2: \omega_1 \omega_3 \omega_2$ and $R_3: \omega_1 \omega_2 \omega_3$. When $n=5$, consider the following preference profile: $R_2: \omega_1 \omega_4 \omega_5...$, $R_3: \omega_2 \omega_5 ...$, $R_4: \omega_1 \omega_3 \omega_5 ...$ and $R_5: \omega_1 \omega_2 ... \omega_5$. When 1 prefers $\omega_1$, agents 2, 3 and 4 will form a cycle and 5 will keep $\omega_5$. When 1 prefers $\omega_2$, 3 and 5 will trade and 5 will get $\omega_3$. When 1 prefers $\omega_3$, agents 1, 3 and 2 will form a cycle, then 5 will trade with 4 and get $\omega_4$. When 1 prefers $\omega_4$, she obtains it by trading with 4, then 5 trades with 2 and obtains $\omega_2$. When 1 prefers $\omega_5$, 5 gets $\omega_1$ by trading with her. When $n$ is odd and higher than $5$, the result obtains by combining the constructions from the even-$n$ case and the $n=5$ case.

    \textit{(iii) $\Rightarrow$ (ii)} Suppose that for each pair of agents $i$ and $j$, there is $R_{-i}$ such that $\mathcal{M}_{i|j}^\phi(R_{-i}) = O$. Then $|\mathcal{M}_{i|-i}(R_{-i})| \geq |\mathcal{M}_{i|j}(R_{-i})| = n$. But as an efficient and group strategy proof assignment rule, $\phi$ is non-autarkic and non-bossy, so by Proposition~\ref{prop:non-autarkic-and-non-bossy}, we know that $|\mathcal{M}_{i|i}(R_{-i})| = |\mathcal{M}_{i|-i}(R_{-i})| = n$. So for each $i$, there is a profile $R_{-i}$ such that $\mathcal{M}_{i|i}(R_{-i}) = O$.
\end{proof}

Conceptually, these results display a paradox. The equivalence of (i) and (ii) justifies the widespread intuition that TTC is a symmetric mechanism, since it awards everyone the opportunity for full freedom, for some combination of others' preferences. But our analysis of transitive power hierarchies in hierarchical exchange mechanisms (proposition~\ref{prop:hierarchical-exchange-menu} and its corollaries) also apply to TTC, showing that such a seemingly horizontal mechanism is in no way immune to freedom inequality and unbalanced power relations. Indeed, the equivalence of (i) and (iii) shows that TTC does generate maximally imbalanced bilateral power relations.

That paradox has an additional interest given the Walrasian market interpretation of TTC: the power and freedom privilege of those who are assigned earlier in the algorithmic interpretation of the mechanism can be equivalently understood, in the market interpretation, as a privilege of the richest agents -- those whose endowment has the highest market value\footnote{Cf. Appendix D of the long working paper version.}. This suggests the existence of a form of power in markets that is based on property, not monopoly, and that persists in competitive markets where the price is set impersonally as a pure market-clearing device.

\subsection*{(Bi-polar) Serial Dictatorship}\label{subsec.sd}
Serially dictatorial assignment rules are arguably the simplest efficient and (group-~)strategy-proof rules. An assignment rule $\phi$ is a Serial Dictatorship if there exists an ordering of agents so that the first agent is assigned their most preferred object, while each subsequent agent in the order is assigned their most preferred among the remaining objects. Equivalently, we can define it as a hierarchical exchange rule corresponding to a TC-mechanism, where one agent is the initial owner of all objects and ownership of objects is inherited according to a pre-specified ordering of agents. 

\begin{definition}[Serial dictatorship as a hierarchical exchange rule]
An assignment rule $\phi$ is a Serial Dictatorship iff it is equivalent to a hierarchical exchange rule for which there exists an ordering of agents --  $i_1,\dots i_n$ -- so that, for any proper submatching $\sigma$, the first of all umatched agents owns all unmatched objects: $\forall \sigma \in \mathcal{S}, x\in \overline{O}_\sigma, i_k\in \overline{N}_\sigma:$
\begin{equation*}
(\forall i_l\in \overline{N}_\sigma : k\leq l) \iff c_\sigma(x)=i_k
\end{equation*}
\end{definition}

\cite{BogomolnaiaDebEhlers2005} introduce Bi-polar Serial Dictatorships as a slightly larger class of assignment rules: a Bi-polar Serial Dictatorship divides initial ownership between at most two agents, allows them to pick their most preferred object among their endowments, or trade with each other. After that, ownership passes on according to a pre-specified ordering of agents as under a Serial Dictatorship. \cite{BogomolnaiaDebEhlers2005} characterise these rules on the domain of weak preferences by strategy-proofness,  Pareto-indifference (vacuously satisfied for strict preferences), non-bossiness, and essential single-valuedness.\footnote{For single valued rules, selections from Bi-polar serially dictatorial rules are characterised by strategy-proofness, Pareto-efficiency and weak non-bossiness.} 

\begin{definition}[Bi-polar Serial dictatorship]
An assignment rule $\phi$ is a Bi-polar Serial Dictatorship  iff it is equivalent to a hierarchical exchange rule for which there exists an ordering of agents --  $i_1,\dots i_n$ -- so that, for the initial, empty, submatching $\sigma=\emptyset$, agents $i_{1}$ and $i_2$ are the owners of all objects, i.e., $c_{\emptyset}(x)\in\{i_1,i_2\}$ for all $x\in O$. For any proper submatching $\sigma\neq \emptyset$, the first of all umatched agents owns all unmatched objects: $\forall \sigma \in \mathcal{S}\backslash \{\emptyset\}, x\in \overline{O}_\sigma, i_k\in \overline{N}_\sigma:$
\begin{equation*}
(\forall i_l\in \overline{N}_\sigma : k\leq l) \iff c_\sigma(x)=i_k\in \overline{N}\sigma
\end{equation*}
\end{definition}

It is straightforward to see that under a (Bi-polar) Serial Dictatorships, agents differ a lot in terms of the influence that they have on their own relevant outcome. However, perhaps surprisingly, and in stark contrast to what Theorem~\ref{thm:ttc-char} has shown for TTC, they uniformly \emph{minimise} the influence that any agent may have over any other individual agent -- by varying their own (reported) preferences, an agent can change the outcome of another agent to at most one other possible object.

\begin{lemma}
\label{lemma:SD-power}
Consider a serial dictatorship $\phi$, two agents $i_k$ and $i_l$ with $k<l$ and a preference profile $R\in \mathcal{R}$. Then $|\mathcal{M}_{i_l,i_k}^\phi(R_{-i_l})| = 1$ and $|\mathcal{M}_{i_k,i_l}^\phi(R_{-i_k})| = 2$.
\end{lemma}

\begin{proof}
Let $\phi$ be a serial dictatorship assigning agents in $N=\{1,...,n\}$ to objects in $O$;  w.l.o.g. let $\phi$ order agents in the natural order $1<2...<n$. Then for any two agents $i < j$ and any preference profile~$R$, we want to show that $|\mathcal{M}_{j|i}^\phi(R_{-j})| = 1$ and $|\mathcal{M}_{i|j}^\phi(R_{-i})| = 2$. The first part of the claim is obvious. Let us prove the second.

First, consider the influence of agent $1$ over $2$. Note that by reporting her preferences $R_1$, and in particular a most preferred-object $x$, she allows $2$ to choose from a subset $O\backslash\{x\}$. As only one object among objects in $O$ is missing, either $2$'s most-preferred or second-most-preferred object is still included, thus there are  two different objects that $2$ will be assigned, depending $1$'s choice. 

We next state this simple fact in more abstract form, allowing us to generalise it beyond agents $1$ and $2$: if there is a collection of subsets of some set $O$, i.e., $\mathcal{O} \subseteq\{O'|O'\subseteq O\}$, such that each $O'\in \mathcal{O}$ is missing at most one object from the union of all subsets, i.e., $|O'|=|\cup_{O''\in \mathcal{O}} O''|-1$, and that every object in the union of all subsets is missing from some subset, i.e., $\cap_{O''\in \mathcal{O}} O''=\emptyset$ then an agent with strict preferences over $\cup_{O''\in \mathcal{O}} O''$ who has to choose from some $O'\in \mathcal{O}$ will be able to pick her most- or second-most-preferred object among all objects in $\cup_{O''\in \mathcal{O}} O''$.

Crucially, as long as there are still $2$ objects to choose from, the resulting new collection of subsets of remaining objects still has the abstract structure above:

\textbf{Claim 1:} \emph{Consider $\mathcal{O}\subseteq\{O'|O'\subseteq O\}$ such that $|O'|=|\cup_{O''\in \mathcal{O}} O''|-1\geq 2$ for all $O'\in  \mathcal{O}$, and $\cap_{O''\in \mathcal{O}} O''=\emptyset$. Moreover, consider two objects $a,b\in \cup_{O''\in \mathcal{O}} O''$ and construct a new collection $\mathcal{O}'$ of subsets where for each subset that contains $a$, we eliminate $a$ and for each that does not contain $a$, we eliminate~$b$. Then $|O'|=|\cup_{O''\in \mathcal{O}'} O''|-1$ for all $O'\in  \mathcal{O}'$, and $\cap_{O''\in \mathcal{O}'} O''=\emptyset$.}

To verify the claim, first observe that the construction of $\mathcal{O}'$ is well defined: since all subsets in $\mathcal{O}$ lack one object compared to their union, either $a$ or $b$ are included in any such subset and we can always eliminate one of them. 

Next, each set in $\mathcal{O}'$ is then smaller by 1, compared to the set in $\mathcal{O}$ from which it was constructed. What remains to show is that $|\cup_{O''\in \mathcal{O}'} O''|$ is also smaller by one compared to $|\cup_{O''\in \mathcal{O}} O''|$. To see this, note that $a\in \cup_{O''\in \mathcal{O}} O''$ while $a$ is not included in any $O''\in \mathcal{O}'$. Moreover since $|\cup_{O''\in \mathcal{O}} O''|-1\geq 2$, there is a third object in the union, say $c$, and since $\cap_{O''\in \mathcal{O}} O''=\emptyset$ there is a set in $\mathcal{O}$ that did not contain $c$ and instead both $a$ and $b$. Thus, for the corresponding set in  $\mathcal{O}'$ constructed from it, we know that it contains $b$. But then $b\in \cup_{O''\in \mathcal{O}'} O''$. Last, observe that $\cap_{O''\in \mathcal{O}} O''=\emptyset$ implies $\cap_{O''\in \mathcal{O}'} O''=\emptyset$ as each set in $\mathcal{O}'$ is a subset of the set in $\mathcal{O}$ from which it was constructed. This proves the claim.

Now, returning to agent $1$, we saw that they can affect $2$'s assignment only between two objects because all the sets that $2$ can potentially choose from are close to each other (each lacking 1 element, relative to the union). But by claim 1, the same also applies to agent 3: the collection of subsets that $3$ could potentially choose from (depending on varying $R_1$, keeping others' preferences fixed, in particular $R_2$) 
are close in the same way, given that $2$ will eliminate either their most-preferred object ($a$) or their second-most-preferred one ($b$). Thus agent 3 will obtain either her most-preferred or her second-most-preferred object in the union of these subsets (note that they are not necessarily her two most-preferred objects in the larger set $O$). Continuing in this way, we see that $1$ can influence the choice of any subsequent agent $j$ between only two objects, as long as agent $j-1$ chose from sets of size at least two -- but this holds for all subsequent agents up until $j=n$.

Finally, the same holds for any $i>1$ in relation to subsequent agents $j>i$ -- simply consider the reduced problem where agents $1,...,i-1$ have been removed together with the objects they chose and relabel the remaining agents. \end{proof}


In fact, Bi-polar Serial Dictatorships are the only efficient and strategy-proof rules where no agent has a greater influence on the outcome of any other agent.
\begin{theorem}
\label{thm:bipolar-SD-char}
Consider an efficient and group-strategy-proof assignment rule $\phi$. The following statements are equivalent:
\begin{enumerate}
    \item  For all preference profiles $R\in \mathcal{R}$ and all $i,j\in N$, we  have $|\mathcal{M}_{i|j}^\phi(R_{-i})| \leq2$.
    \item $\phi$ is a Bi-polar Serial Dictatorship.
\end{enumerate}
\end{theorem}

The theorem is proven in Appendix~\ref{appendix:proof-bipolar-SD-char}. While the direction $(2) \Rightarrow (1)$ derives from lemma~\ref{lemma:SD-power}, the other direction proceeds as follows. We consider a TC-mechanism with an associated control right structure that implements $\phi$ and show that the constraint on bilateral power implies that there is no submatching where more than two agents can control objects (claim 1) and that we can assume, w.l.o.g. that neiter of them is a broker (claim 2). We then consider the case that there is one owner initially to show that at the next submatching there needs to be a single owner (claim 4) and that the preferences of the initial owner do not affect the order of inheritance among subsequent owners (claim 5). We then show the analogue for the remaining case of two initial owners.

\section{Conclusion}
In this paper, we have proposed to measure an agent's influence in a strategy-proof mechanism by their menu, the set of possible outcomes that they can bring about by varying their strategy, i.e., their (reported) preferences. Further, the approach allows to disaggregate influence into freedom -- an agent's influence on her own relevant outcome -- and power over others -- influence on their relevant outcome -- and encompasses as special cases the ranking of opportunity sets from the freedom of choice literature and power indices from the literature on binary voting.

As our first main result, we uncover a subtle connection between welfare and freedom: while the two notions do not coincide (an agent may find their preferences satisfied because of “luck” while lacking any influence), we show that, under a mild richness condition on preferences, a mechanism Pareto-dominates another mechanism if and \emph{only if} it awards each agent with more freedom, i.e., larger menus in a set inclusion sense. Hence, the set of constrained efficient mechanisms coincides with the set of mechanisms that offer a maximal degree of freedom. 

Second, we fruitfully apply our approach to assignment rules. We describe power and freedom in group-strategy-proof and efficient rules and, within that class, uncover a stark trade-off between two (arguably normatively appealing) properties: first, universal possibility of complete freedom, i.e., allowing everyone to choose among all objects at some opposing preference profile; second, minimal bilateral power, i.e., ensuring that no one can influence another individual's outcome beyond what is necessary. The top trading cycle possesses the first property (and indeed, is characterised by it among group-strategy-proof and efficient rules), but generates situations where an agent has \emph{maximal} bilateral power on another. Bipolar serial dictatorship mechanisms possess the second property (and indeed, are characterised by it in the same class of rules), but rigidly predetermine a freedom hierarchy. 

We hope that our framework can serve as a first step towards a unified game-theoretic analysis of power and freedom in economic environments, meeting the challenge from political philosophers. In particular, we would hope to see our approach extended beyond direct, dominant-strategy mechanisms, taking into account other players' equilibrium strategy adaptation to a player's change in strategy. Among other things, this would allow applications to many industrial organisation models, and thus the analysis of monopoly power.

\section*{Appendix}
\appendix

\section{Power and preferences}
\label{appendix:power-preferences}
In our framework, the players' preferences come into play at two distinct levels when defining the menu of agent $i$. On the one hand, we take the actual preferences of all other players as given; on the other hand, to test the effect of various counterfactual possibilities for $i$'s own preferences. So $i$'s power depends on the other players' preference profile, and on $i$'s preference domain. It is worthwhile to review in this light some existing debates in game theory on the relation between preferences, power and freedom.

\paragraph{Actual versus \textit{a priori} voting power} \citet{FelsenthalMachover1998, FelsenthalMachover2004} proposed two influential distinctions. The first is between I-power (power as influence, or ‘‘a voter’s potential influence over the outcome" \citep[p. 9]{FelsenthalMachover2004}) and P-power (power as payoff, or ‘‘a voter's expected relative share in a fixed prize"). The framework proposed here, which generalises the notion of decisiveness, belongs to the first category: it measures the influence of an agent over the outcome, not her welfare.

The second distinction proposed by Felsenthal and Machover contrasts  \textit{a priori} (‘‘the component of the [players]’ voting power that derives solely from the decision rule itself", ibid., p.13) from \textit{actual} voting power (which ‘‘depends on a complex interaction of real-world factors", esp. ‘‘diplomacy, political pressures, members’ specific interests and preferences", ibid.). They argue that the complexity of factors affecting actual voting power makes its analysis impracticable, especially for normative purpose, and hence they choose to focus on \textit{a priori} power, i.e. power indices derived from the voting procedure alone, independently of strategic interactions and of players' preferences.

They go even further, arguing that when trying to take all actual forces into account, ‘‘we arguably move away from considering power altogether: we move from ‘who can get what’ to ‘who does get what’" (ibid.:14), i.e. from I to P-power. We believe that our framework is able to meet precisely this challenge. It allows for actual power measurement, taking into account how other players' preferences and strategic interactions shape an agent's opportunities, while also preserving a dimension of potentiality or causality. As argued above, the key step to analyse an agent's power in that way is to take only the other players' preferences as given, while varying her own preferences. This variation preserves a counterfactual element in the analysis of actual power, hence preventing a confusion between ‘‘who can get what" and ‘‘who gets what".

Note, furthermore, that the interest in actual power does not mean the abandonment of \textit{a priori} power. Instead, what Felsenthal and Machover call \textit{a priori} power can be obtained from the expected cardinality of the menu's size, under some assumption for the probability distribution of the preferences of the other agents. In particular, in binary voting games, an agent's menu is either of size 2 or 1, depending on whether she is decisive or not; the classic \textit{a priori} power indices then obtain as expectations of that menu's cardinality under natural assumptions for the distribution of preferences for the other players.

\paragraph{Game versus game-form} A similar rivalry between \textit{a priori} and actual power analysis has expressed itself forcefully in the debate between \citet{NapelWidgren2004, BrahamHoller2005, NapelWidgren2005, BrahamHoller2005b}. In the introduction, we mentioned Napel and Widgren's framework, which is one of the inspirations for this paper. Braham and Holler objected to it. Their argument is that the meaning of power is an ‘‘ability to affect outcomes", not the effect itself. Therefore, according to them, the measurement of power can not refer to the preferences of the agents themselves. Power is a property of a game form, the rules of the game, not the game itself, which includes the preferences of the agents. (See \citet{Bervoets2007} for a similar argument regarding the measurement of freedom.)

\citet{BrahamHoller2005}'s main argument is semantic: they claim that power is a causal or potential category and hence should be deduced from the analysis of the game form, and not be conflated with a claim about actual outcomes that are derived from the preferences of the agents. But in the framework presented here, the power of an agent depends on the preferences of others, not her own. By varying her own preferences in her preference domain, the principle of causality or potentiality is preserved; while at the same time, by taking the preferences of others as given, one makes sure that real as opposed to formal power is analysed, and takes into account the effective forces that constrain an agent.

\section{Proof and corollaries of proposition~\ref{prop:delta-rank}}
\label{appendix:delta-rank}
To prove the proposition, we start with the following lemma.
\begin{lemma}
\label{lemma:rho-conditional-delta}
Consider an environment $(A,N, \mathcal{R})$, an agent  $i\in N$ such that $|A_i| = m$ and $\mathcal{R}_i$ is the set of strict orderings of $A_i$, and a strategy-proof mechanism $\phi$. Suppose that the preference profile $R$ is drawn at random so that $R_i$ is uniform in $\mathcal{R}_i$ and independent from $R_{-i}$. Then for all $(r,s) \in \{1,..., m-1\}^2$: 
  \[
        \mathbb{P}[\rho_i^\phi(R) > r | |\mathcal{M}_{i|i}^\phi(R_{-i})| = s] = 
        \begin{cases}            
       \frac{(m-r)!(m-s)!}{m!(m-r-s)!} & \text{if } r +s \leq m \\
        0 & \text{otherwise}
         \end{cases} 
  \]
\end{lemma}
\begin{proof}
Write $\mathcal{M}_{i}$, $|\mathcal{M}_{i}|$ and $\rho_i$ for $\mathcal{M}_{i|i}^\phi(R_{-i})$, $|\mathcal{M}_{i|i}^\phi(R_{-i})|$ and $\rho_i^\phi(R_i,R_{-i})$. $R_i$ is a strict ordering;  let $R_{i,1}$ denote its top-ranked element, $R_{i,2}$ the element it ranks second, etc.

There are $|\mathcal{M}_{i}|$ elements of $A_i$ available to $i$. By strategy-proofness, all elements of rank strictly less than $\rho_i$ in $i$'s preference list are not available to $i$, so $\rho_i -1 + |\mathcal{M}_{i}| \leq m$. So when $r+s > m$, $\mathbb{P}[\rho_i > r | |\mathcal{M}_{i}| = s] = 0$.

Now suppose $r+s \leq m$. $R_i$ is independent from $R_{-i}$, so independent from $|\mathcal{M}_{i}(R_{-i})|$, and the law of $R_i$  conditional on $\{ |\mathcal{M}_{i}| = s\}$ is the same as its unconditional law, i.e., the uniform distribution in $\mathcal{R}_i$. Under that distribution, $R_{i,1}$ is uniform in $A_i$; conditional on $R_{i,1}$, $R_{i,2}$  is uniform in $A_i \setminus \{R_{i,1} \}$; etc.

Then we can write:
\begin{align*}
    \mathbb{P}[\rho_i > r \; | \; |\mathcal{M}_{i}| = s] & = \mathbb{P}[R_{i,1} \notin \mathcal{M}_{i}, \; R_{i,2} \notin \mathcal{M}_{i}, ..., R_{i,r} \notin \mathcal{M}_{i} \; | \; |\mathcal{M}_{i}| = s]\\
    & = \mathbb{P}[R_{i,1} \notin \mathcal{M}_{i} \; | \; |\mathcal{M}_{i}| = s] \times \mathbb{P}[R_{i,2} \notin \mathcal{M}_{i} \; | \; R_{i,1} \notin \mathcal{M}_{i}, |\mathcal{M}_{i}| = s] \\
    & \; \; \; \; \; \; \times ...\times \mathbb{P}[R_{i,r} \notin \mathcal{M}_{i} | R_{i,1} \notin \mathcal{M}_{i}, ..., R_{i,r-1} \notin \mathcal{M}_{i}, |\mathcal{M}_{i}| = s] \\
    & = \frac{m-s}{m} \frac{m-s-1}{m-1} ... \frac{m-s-r+1}{m-r+1} \\
    & = \frac{(m-r)!(m-s)!}{m!(m-r-s)!}
\end{align*}
\end{proof}

This allows us to prove proposition~\ref{prop:delta-rank}.

    Indeed, the case where $R_{-i}$ is deterministically given is a special case of when it is random.
    Thus, by lemma~\ref{lemma:rho-conditional-delta} we have 
     \[
        \mathbb{P}[\rho_i^\phi(R) > r  ] = 
        \begin{cases}            
       \frac{m-s}{m} \frac{m-s-1}{m-1} ... \frac{m-s-r+1}{m-r+1} & \text{if } r +s \leq m \\
        0 & \text{otherwise}
         \end{cases} 
  \]
  with $s := |\mathcal{M}_{i|i}^\phi(R_{-i})|$ and $m :=|A_i|=m$.
    Now suppose $|\mathcal{M}_{i|i}^\phi(R_{-i})| \geq |\mathcal{M}_{j|j}^\psi(R'_{-j})|$ and $|A_i|=m=|A_j|$. Then because the right hand side is weakly decreasing in $s$, for all $r \geq 1$, $\mathbb{P}[\rho_i > r] \leq \mathbb{P}[\rho_j > r]$, i.e., $\rho_j$ first-order stochastically dominates $\rho_i$.

    Conversely, suppose that $\rho_j$ first-order stochastically dominates $\rho_i$.  Then $\mathbb{P}[\rho_i > 1] \leq \mathbb{P}[\rho_j > 1]$. But for $r=1$, lemma~\ref{lemma:rho-conditional-delta} implies $\mathbb{P}[\rho_i > 1] = 1 - \frac{|\mathcal{M}_{i}|}{m}$ and likewise for $j$. So $|\mathcal{M}_{i|i}^\phi(R_{-i})| \geq |\mathcal{M}_{j|j}^\psi(R'_{-j})|$. \qed

Can this result be generalised when others’ preferences are random too? Elements of an answer are provided by the following two corollaries, both under the assumption of $R_{i}$ (resp. $R'_{j})$ uniformly random and independent from $R_{-i}$ (resp. $R'_{-j}$) -- the “impartial culture” assumption often invoked in voting theory.

\begin{corollary}
\label{cor:delta-rank}
If $|\mathcal{M}_{i|i}^\phi(R_{-i})|$ first-order stochastically dominates $|\mathcal{M}_{j|j}^\psi(R'_{-j})|$, then $\rho_j^\psi(R')$ first-order stochastically dominates $\rho_i^\phi(R)$.
\end{corollary}
\begin{proof}
    For a given $r$,
$\mathbb{P}[\rho_i > r] = \mathbb{E}\left[ \mathbb{P}[\rho_i > r | |\mathcal{M}^\phi_{i|i}(R_{-i})|] \right]$.
    
    Now suppose that $|\mathcal{M}_{i|i}^\phi(R_{-i})|$ first-order stochastically dominates $|\mathcal{M}_{j|j}^\psi(R'_{-j})|$. We know from its explicit form given by lemma~\ref{lemma:rho-conditional-delta} above that $\mathbb{P}[\rho_i > r | |\mathcal{M}^\phi_{i|i}(R_{-i})|]$ is a weakly decreasing function of $|\mathcal{M}^\phi_{i|i}(R_{-i})|$ on $\mathbb{R}_+$. Therefore, $\mathbb{E}\left[ \mathbb{P}[\rho_i > r | |\mathcal{M}^\phi_{i|i}(R_{-i})|] \right] \leq \mathbb{E}\left[ \mathbb{P}[\rho_j > r | |\mathcal{M}^\psi_{j|j}(R_{-j})|] \right]$. This is true for any $r$, which yields the desired result.
\end{proof}

\begin{corollary}
\label{cor:probability-preferred-outcome}
    \begin{equation*}
        \mathbb{E}[|\mathcal{M}_{i|i}^\phi(R_{-i})|] = | A_i|\times \mathbb{P}[\rho_i^\phi(R) = 1]
    \end{equation*}
\end{corollary}
\begin{proof}
    Applying lemma~\ref{lemma:rho-conditional-delta} with $r=1$ gives
    $\mathbb{P}[\rho_i > 1 | |\mathcal{M}_{i}| = s] = \frac{m-s}{m}$. Averaging over $s$ gives the desired result.
\end{proof}

Finally, another technical corollary derives from lemma~\ref{lemma:rho-conditional-delta}.
\begin{lemma}
\label{lemma:rho-delta-linear-bijection}
Consider $(A, N, \mathcal{R})$ an environment, $\phi$ a mechanism, $i$ an agent such that $|A_i| = m$ and $\mathcal{R}_i$ the set of strict orderings of $A_i$. Suppose that $R$ is randomly drawn, so that $R_{i}$ is uniform in $\mathcal{R}_i$ and independent from $R_{-i}$. Then the following linear, bijective relation holds between the probability distribution of $\rho_i$ and that of $|\mathcal{M}_i|$:
\begin{equation*}
    (\mathbb{P}[\rho_i > r ])_{1 \leq r \leq m-1} = M(\mathbb{P}[|\mathcal{M}_i| = s])_{1\leq s \leq m-1}
\end{equation*}
where $M$ is the $(m-1)\times (m-1)$ matrix defined by $M_{r,s} = \frac{(m-r)!(m-s)!}{m!(m-r-s)!}$ when $r+s \leq m$, 0 otherwise.
\end{lemma}
\begin{proof}
    \begin{align*}
        \mathbb{P}[\rho_i > r] & = \sum_{s = 1}^m \mathbb{P}[|\mathcal{M}_i| = s] \mathbb{P}[\rho_i > r \; | \; |\mathcal{M}_i| = s] & \\
        & = \sum_{s=1}^{m-r} \mathbb{P}[|\mathcal{M}_i| = s] \frac{(m-r)!(m-s)!}{m!(m-r-s)!} & \text{by lemma~\ref{lemma:rho-conditional-delta}}
    \end{align*}
    That matrix is anti-triangular, so invertible, so the relation is bijective and $(\mathbb{P}[|\mathcal{M}_i| = s])_{1\leq s \leq m-1} = M^{-1} (\mathbb{P}[\rho_i > r ])_{1 \leq r \leq m-1}$.
\end{proof}

\section{Binary games}
\label{appendix:binary-voting}
In the following, $\phi$ is a mechanism with two possible outcomes $A = \{ 0, 1 \}$, and such that all $n$ agents share the same preference domain of two strict orderings, the preference of 0 over 1, denoted 0, and the opposite one, denoted 1.

Given $R_{-i}$, agent $i$ is said \textit{decisive} if the outcome varies with her report: $\phi(R_i = 0, R_{-i}) \neq \phi(R_i = 1, R_{-i})$. Then it is clear that given $R_{-i}$, agent $i$ is decisive if and only if she has two possible outcomes in her menu, i.e. $|\mathcal{M}^\phi_{i|i}(R_{-i})| = 2$; otherwise $|\mathcal{M}^\phi_{i|i}(R_{-i})| = 1$. 

Let us denote $\mathcal{B}_i^\phi := \{ R_{-i} \in \{0,1\}^{n-1} | i \text{ is decisive when others have preferences }R_{-i}\}$ the set of others' preferences that make $i$ decisive.

\begin{definition}[Banzhaf index \citep{Banzhaf1965}]
An agent's Banzhaf index $b^\phi_i$ is the proportion of preference profiles for the other players under which $i$ is decisive:
\begin{equation*}
    b^\phi_i = \frac{|\mathcal{B}_i^\phi |}{2^{n-1}}
\end{equation*}
\end{definition}

If $R_{-i} \in \{0,1\}^{n-1}$, let us denote $\sigma(R_{-i}) = \sum_{j \neq i} R_j$, the number of others who prefer 1 over 0. Then we can define the other famous power index.

\begin{definition}[Shapley-Shubik index \citep{ShapleyShubik1954}]
Under mechanism $\phi$, $i$'s Shapley-Shubik index $s_i^\phi$ is defined as:
\begin{equation*}
   s_i^\phi = \sum_{R_{-i} \in \mathcal{B}_i^\phi} \frac{\sigma(R_{-i})!(n-1-\sigma(R_{-i}))!}{n!}
\end{equation*}
\end{definition}

\begin{proposition}[Expected power \citep{Straffin1988}]
\label{prop:Straffin}
    If the $R_j$s ($j \neq i$) are i.i.d. Bernoulli variables of parameter $1/2$, then
    \begin{equation*}
        \mathbb{P}[i \text{ is decisive under $\phi$ when others have preferences } R_{-i}] = b_i^\phi
    \end{equation*}
    
    If $p$ is a random variable uniformly drawn in $[0,1]$ and, conditional on $p$, $R_{-i}$ are i.i.d. Bernoulli variables of parameter $p$, then
    \begin{equation*}
        \mathbb{P}[i \text{ is decisive under $\phi$ when others have preferences } R_{-i}] = s_i^\phi
    \end{equation*}
\end{proposition}
But we also know that
\begin{equation*}
E[|\mathcal{M}_{i|i}^\phi(R_{-i})|] = \mathbb{P}[i \text{ is decisive under $\phi$ when others have preferences } R_{-i}] + 1.
\end{equation*}
So under the probability distributions for others' preferences suggested by Straffin, there is a direct relation between the expectation of an agent's cardinal freedom index and her Banzhaf, or her Shapley-Shubik index.

\section{Additional results on TTC}
\label{appendix:other-results-ttc}
When they introduced this rule, \citet{ShapleyScarf1974} already noticed that it is equivalent to a Walrasian market allocation.

\begin{definition}[Equilibrium prices]
    Consider a TTC-rule $\phi$ with initial endowments $\omega=(\omega_i)_{i\in N}$ and a preference profile $R$. Then a \textit{price vector supporting the allocation $\phi(R)$} is a vector $p \in {\mathbb{R}_+^0}^n$ such that for each $i\in N$, $\phi_i(R)$ is the $R_i$ most preferred object in $i$'s associated budget set $B_i(p):=\{x\in O|p_x\leq p_{\omega_i}\}$. Let $\mathcal{P}^\phi(R)$ denote the set of all price vectors supporting $\phi(R)$.
\end{definition}

We may connect a price vector $p$ supporting $\phi(R)$ and the TTC-algorithm as follows: to allow agents within a cycle to trade, assign the same price to all objects within a cycle; to prevent agents within cycles that are formed later to demand objects that are in cycles formed earlier, assign to the former a lower price than to the latter. Thus, if there is a sequence of cycles $C_1, ..., C_m$ consistent with the TTC-algorithm\footnote{Recall that in general, there are multiple ordered sequence of cycles compatible with the TTC-algorithm, because of the indeterminacy that arises when multiple cycles arise simultaneously.} we may choose $p$ such that
    \begin{equation*}
        \forall i \in C_a, j \in C_b, \; p_{\omega_i} > p_{\omega_j} \Leftrightarrow a < b.
    \end{equation*}
Conversely, if $p_{\omega_i}\geq p_{\omega_j}$ for some supporting price vector, there is an order in which cycles can be cleared such that $a\leq b$ when $C_a$ contains $i$ and $C_b$ contains $j$.

Because of this remarkable connection, before providing our characterisation of the TTC rule, it may be interesting to interpret our previous results on power and freedom in hierarchical exchange in terms of wealth, i.e. the equilibrium price of an agent's endowment.

\begin{observation}
\label{obs:freedom-prices}
If $\phi$ is a TTC assignment rule and $i$, $j$ are two agents, then:
    \begin{equation*}
        \mathcal{M}^\phi_{i|i}(R_{-i}) \subseteq \mathcal{M}_{j|j}^\phi(R_{-j})\Leftrightarrow \forall p \in \mathcal{P}^\phi(R), p_{\omega_i} \leq p_{\omega_j}
    \end{equation*}
\end{observation}

\begin{proof}
\begin{align*}
    \mathcal{M}^\phi_{i|i}(R_{-i}) \subseteq \mathcal{M}_{j|j}^\phi(R_{-j}) & \Leftrightarrow \text{under $R_{-j}$, $j$ has power over $i$} & \text{by proposition~\ref{prop:hierarchical-exchange-menu}} \\
    & \Leftrightarrow \text{under $R$, for any order of the removal} & \\
    & \hspace{1cm} \text{of cycles, $i$ is matched weakly after $j$} & \\
    & \Leftrightarrow \forall p \in \mathcal{P}^\phi(R), p_{\omega_i} \leq p_{\omega_j} & \text{ as stated above}
\end{align*}
\end{proof}
Thus, an agent $j$ has more freedom to choose than another agent $i$ if, at the given preference profile, \emph{all} price vectors supporting the TTC-allocation ascribe a higher wealth to $j$ than to $i$.

\begin{corollary}
\label{cor:TTC-Menus-vs-Budgetsets}
For any preference profile $R$,
    \begin{equation*}
        \mathcal{M}^\phi_{i|i}(R_{-i}) = \cap_{p \in \mathcal{P}^\phi(R)}B_i(p).
    \end{equation*}
\end{corollary}

\begin{proof}
    Consider object $x \in \mathcal{M}^\phi_{i|i}(R_{-i})$ and denote $j$ the agent endowed with $x$, i.e. $x = \omega_j$. Then $i$ has power over $j$, and $j$ is always matched after $i$. So, considering $p$ a price vector that supports the allocation $\phi(R)$, we have $p_{\omega_j} \leq p_{\omega_i}$. So $\omega_j \in B_i(p)$. So $\mathcal{M}_{i|i}(R_{-i}) \subseteq \cap_{p \in \mathcal{P}^\phi(R)} B_i(p)$.

    Conversely, consider $\omega_j \in \cap_{p \in \mathcal{P}^\phi(R)} B_i(p)$, i.e. $\forall p \in \mathcal{P}^\phi(R), p_{\omega_j} \leq p_{\omega_i}$. Then $j$ is always matched after $i$, so $i$ has power over $j$, so $\mathcal{M}_{i|i}^\phi(R_{-i}) \supseteq \mathcal{M}_{j|j}^\phi(R_{-i})$, so $\omega_j \in \mathcal{M}_{i|i}^\phi(R_{-i})$.
\end{proof}

Corollary \ref{cor:TTC-Menus-vs-Budgetsets} connects two related concepts: the notion of menus and the notion of budgets sets from which agents may choose in a preference maximising manner to arrive at the allocation $\phi(R)$. When choosing from budget sets, we implicitly assume that agents take prices as given and hence neglect the effect of changing their demand -- akin to a change in (reported) preferences -- on the price. In contrast, an agents' menu depends, by construction, only on others' reported preferences and is thus robust to changes of own reported preferences. To identify the analogously robust elements of agents' budget sets, we have to take the intersection of all possible budget sets that may support the allocation $\phi(R)$.

We conclude this Appendix by noting that our characterisation of TTC (Theorem~\ref{thm:ttc-char}) resembles  Theorem 1 in \cite{LongVelez2021} who require assignment rules to be \textit{balanced}, i.e., to yield each agent, for any $k\leq n$, equally often their $k$-most preferred object when considering all possible preference profiles. For $4$ agents and more, TTC-rules are then the only efficient, balanced, and group-strategy-proof assignment rules.\footnote{For $3$ agents the $3$-broker mechanism jointly satisfies these properties as well; Theorem \ref{thm:ttc-char} could be extended to 3 agents analogously.} While balancedness requires that, on average across the space of all preference profiles, agents are treated symmetrically, our requirement does not impose symmetric treatment -- in particular, while we ask that each agent should be in the fortunate position of having unrestricted freedom of choice (so that they receive their most preferred object whatever that may be) at least once,\footnote{Recall that by Proposition \ref{prop:assignment-extreme-situations} there is at least one agent who is in that situation at least once.} we do not impose that different agents should be equally often in this situation. Further, we do not place any restrictions on how often each agent should receive their $k$-most preferred object for $k\geq 2$. Indeed, \citet{LongVelez2021}'s result can be obtained as a corollary of our theorem.

\begin{corollary}
\label{cor:ttc-balanced}
    If an efficient and group-strategy-proof assignment rule $\phi$ with four or more agents is balanced, then it is a TTC rule.
\end{corollary}
\begin{proof}
    Consider $\phi$ an efficient, group-strategy-proof assignment rule with $n$ agents and objects that is balanced. Because $\phi$ is efficient and strategy-proof, by proposition~\ref{prop:assignment-extreme-situations} there is an agent $i$ and an opposing profile $R^0_{-i}$ where $|\mathcal{M}_{i|i}^\psi(R^0_{-i})| = n$.

    Now suppose that a preference profile is drawn following the impartial culture distribution. A profile $R$ such that $R_{-i} = R^0_{-i}$ has a positive probability, so $\mathbb{P}[|\mathcal{M}_{i|i}^\phi(R_{-i})| = n] > 0$. Consider another agent $j$. Because $\phi$ is balanced, $\rho_j^\phi(R)$ has the same distribution as $\rho_i^\phi(R)$. But by lemma~\ref{lemma:rho-delta-linear-bijection} (Appendix~\ref{appendix:delta-rank}), this implies that $|\mathcal{M}_{j|j}^\phi(R_{-j})|$ also has the same distribution as $|\mathcal{M}_{i|i}^\phi(R_{-i})|$. In particular $\mathbb{P}[|\mathcal{M}_{j|j}^\phi(R_{-j})| = n] = \mathbb{P}[|\mathcal{M}_{i|i}^\phi(R_{-i})| = n] > 0$.

    So there exists $R^1$ such that $\mathcal{M}_{j|j}^\phi(R^1_{-j}) = O$. This is true for any $j$, and we have assumed that $\phi$ is efficient and group strategy-proof, so we can apply Theorem~\ref{thm:ttc-char} and $\phi$ is a TTC rule.
\end{proof}

\section{Proof of Theorem~\ref{thm:bipolar-SD-char}}
\label{appendix:proof-bipolar-SD-char}

    $(2) \implies (1)$: Consider a Bi-polar serial dictatorship $\phi$ with an ordering of agents from $i_1$ to $i_n$.  

    Consider first the power of $i_1$ over $i_2$ at some preference profile $R$. Let us denote $i_2$'s most- and second-most-preferred objects by $a$ and $b$ respectively. Then if $i_2$ initially owns $a$, $\mathcal{M}_{i_1|i_2}^\phi(R_{-i_1}) = \{ a \}$, while if $a$ is initially owned by $i_1$ we have $\mathcal{M}_{i_1|i_2}^\phi(R_{-i_1}) = \{ a, b \}$. Similarly, $i_2$'s power over $i_1$ is bounded $|\mathcal{M}_{i_2|i_1}^\phi(R_{-i_2})|\leq 2$.
    
    Next, consider the power of $i_1$ over $i_{l}$, with $l > 2$. Suppose that $i_1$ initially owns $a$, the most-preferred object of $i_2$. Then $i_1$ can pick any object by reporting it as her most-preferred, while $i_2$ receives her most-preferred among the remaining objects. Hence, whatever $i_1$'s preferences, under opposing preferences $R_{-{i_1}}$ the resulting allocation will coincide with that of a standard, unipolar serial dictatorship. Thus, by lemma~\ref{lemma:SD-power}, $|\mathcal{M}_{i_1|i_l}^\phi(R_{-i_1})|\leq 2$.
    If instead, $a$ is owned by $i_2$ herself, $i_2$ receives $a$ while $i_1$ receives her most-preferred among the remaining objects. Hence, $i_1$ is in the same position as  the second ranked agent in a (unipolar) serial dictatorship -- again, by lemma~\ref{lemma:SD-power}, her power over subsequent agents is bounded by $|\mathcal{M}_{i_1|i_l}^\phi(R_{-i_1})|\leq 2$. Analogously, we find that $i_2$'s power over any $i_l$, $l>2$, is bounded in the same way.

    Last, consider the power of agent $i_k$ over $i_l$ for $k>2$. If $l<k$, the assignment of $i_l$ does not depend on $R_{i_k}$ while if $l>k$, the situation is again analogous to a serial dictatorship -- thus $|\mathcal{M}_{i_k|i_l}^\phi (R_{-i_k})|\leq 2$ according to lemma~\ref{lemma:SD-power}.

   $(1) \implies (2)$: For the other direction, suppose that under $\phi$ we have $|\mathcal{M}_{i|j}^\phi(R_{-i}) | \leq 2$ for all $i,j$ and $R$. We want to show that $\phi$ can be represented as a TC-mechanism with a control rights structure $c$ that (i) has at most two owners at the empty submatching, (ii) exactly onw owner at each subsequent proper submatching, and (iii) the order of owners  is fixed.

   \noindent \textbf{Claim 1}: \emph{At any submatching $\sigma$, at most two distinct agents control objects.}

    \emph{Proof.} Suppose not. In the absence of a broker, there are (at least) three unassigned objects $a,b,c$ with three distinct owners $c_\sigma(a)=1$, $c_\sigma(b)=2$, $c_\sigma(c)=3$. If their preferences are given as $R_1:a...$, $R_2:acb...$, and $R_3:abc...$, then $1$ keeps $a$ while $2$ and $3$ trade -- in particular, $3$ is assigned $b$. If instead agent $1$'s preferences are given as $R'_1:b...$, $1$ trades with $2$ while $3$ keeps her object $c$. Last, if $1$'s preferences are given as $R''_1:c...$, $1$ and $3$ trade and $3$ is assigned $a$. Thus, depending on the preferences of $1$, $3$ is assigned at least three different objects, $a$, $b$, or $c$, i.e., $|\mathcal{M}_{i|j}^\phi(R_{-i})| >2$ (for $i=1$ and $j=3$).
  
   Alternatively, there may be a broker and (at least) two owners. 
   To derive a contradiction, let $3$ be the broker for $c$; 
   let $1$ and $2$ still own $a$ and $b$ and let $R_{-1}$ be unchanged.
   Then as above, we again find that depending on the preferences of $1$, $3$ is assigned at least three different objects, $a$, $b$, or $c$. 
   

    As last remaining case, suppose there are three brokers and three objects -- say $1$ is broker for $a$, $2$ for $b$ and $3$ for $c$. Moreover, suppose that $R_2:bca$ and $R_3:cba$. Then if $1$ ranks $b$ first, the resulting assignment is $(\mu_1,\mu_2,\mu_3)=(b,c,a)$ as it is the only Pareto-efficient assignment that avoids assigning objects to their broker. If instead $1$ ranks $c$ first the resulting assignment is $(c,a,b)$, by the same logic. Last, if $1$ ranks $a$ first, there is only one Pareto-efficient assignment, namely $(a,b,c)$. Thus depending on the preferences of $1$, each of the other two agents is assigned three different objects.\hfill{$\diamond$}

   \noindent \textbf{Claim 2}: \emph{Without loss of generality, we can assume that the control right structure $c$ that implements $\phi$ via a TC-mechanism is such that there is no broker.}
   
   \emph{Proof.} 
    Towards a contradiction, suppose that every TC-mechanism with control rights structure $c$ that implements $\phi$ includes a submatching  $\sigma$ where some agent $i$ is a broker for some object $x$. By claim 1, all remaining unassigned objects $\overline{O}_\sigma\backslash\{x\}$ are then owned by a single agent $j\neq i$. Since there are finitely many submatchings, we can further assume that $\sigma$ is minimal, i.e., that at all submatchings $\sigma'$ with $|\sigma'|<|\sigma|$, there is no broker.

    Depending on the preferences of $i$ and $j$, the next submatching is either $\sigma^y=\sigma\cup\{(j,y)\}$ if $j$ considers some $y\in\overline{O}_\sigma\backslash\{x\}$ most-preferred or $\sigma^{xy} =\sigma\cup\{(j,x),(i,y)\}$ if $j$ considers $x$ most preferred and $y$ is the most preferred object of $i$ among objects other than $x$. In particular the submatching $\sigma^x=\sigma\cup\{(j,x)\}$ is never reached from $\sigma$, given $c$.

   But then we can define an alternative control right structure $c'$ where at $\sigma$, $j$ owns all objects (including $x$), and at $\sigma^{x}$, $i$ is the owner of all objects other than $x$. At all other submatchings $\sigma'$, we set $c'_{\sigma'}=c_{\sigma'}$. 

      To see that $c'$ is a well-defined control rights structure, and leads to the same assignments as $c$, first consider any $\sigma'$ with $|\sigma'|<|\sigma|$. As $\sigma$ was a minimal submatching with a broker, we know that at $\sigma'$ there is no broker and at most two owners (claim 1). Moreover, as $c'_{\sigma'}=c_{\sigma'}$ the same objects will be assigned to current owners under $c$ and $c'$ and we move on to the next submatching $\sigma''$ where $|\sigma''|\leq |\sigma'|+2\leq |\sigma|+1=|\sigma^x|$. If at $\sigma''$ $c$ and $c'$ agree as well, then control rights were preserved under $c'$ as they were under $c$. If $c$ and $c'$ differ at $\sigma''$, that implies (i) $\sigma''=\sigma$ or (ii) $\sigma''=\sigma^x$. 
      
      Consider (i) $\sigma''=\sigma$ : every agent who was an owner under $c$, namely $j$, is still is an owner under $c'$, thus their ownership was preserved under $c'$, given that this was the case under $c$. 
      
      Consider (ii) $\sigma''=\sigma^x$: then we moved from $\sigma'$ with $|\sigma'|<|\sigma|$ to $|\sigma''|>|\sigma|$ which implies that there were two owners at $\sigma'$ who were both matched -- and there are no ownership rights to be preserved.

      Second, suppose $\sigma'=\sigma$. Then under both $c$ and $c'$, $j$ is assigned her most-preferred object. If that object is $z\neq x$, we move to $\sigma''=\sigma \cup \{(j,z)\}$, both under $c$ and $c'$, where $i$ is the broker for $x$ (still serving in that role under $c$, new to that role under $c'$). Since the only owner at $\sigma'$ was assigned, there are no onwership-rights to be preserved. If instead $j$ considers $x$ most-preferred, we arrive at  $\sigma^{xy}$ under $c$; under $c'$ we arrive at the same submatching but in a separate step via $\sigma^{x}$.

   Last, suppose $|\sigma'|\geq|\sigma|$ with $\sigma'\neq \sigma$. But then there is some agent $k\neq i,j$ assigned under $\sigma'$ but neither under $\sigma$, nor $\sigma^x$. Thus, $c$ and $c'$ agree for $\sigma'$ and yield the same subsequent submatching $\sigma''$. As $k$ is still assigned under a subsequent submatching $\sigma''$, we know that $\sigma''\neq \sigma, \sigma^x$, hence $c$ and $c'$ also agree for $\sigma''$. But then $c'$ preserves control rights in the same way as $c$.

   This establishes $c'$ as a well defined control right structure that yields the same assignments as $c$. Since $c'$ has strictly fewer submatchings with a broker, and as there are finitely many submatchings to consider, repeated application of the argument allows us to eliminate any broker, thus establishing the claim. \hfill{$\diamond$}

   By our first two claims, we can assume that there is either a single owner or two owners at any submatching  $\sigma$. We first consider the single owner case and show that their choice does not influence who inherits an unpicked object (claim 3), that whatever they pick there is only one agent who inherits all unpicked objects (claim 4), that the identity of this agent is not influenced by the object that the current owner picks (again claim 4), and that the choice of the current owner also does not influence inheritance at a later point (claim 5). 

   \noindent \textbf{Claim 3}:  \emph{Consider a submatching $\sigma$, where $i$ is the single owner of at least three unassigned objects, including $z$. Then it cannot be that when $i$ chooses some object, say $x$, $j$ becomes the owner of $z$ at $\sigma^{x}=\sigma\cup\{(i,x)\}$ while when $i$ chooses another object, say $y$, another agent $k\neq j$ becomes the owner of $z$ at $\sigma^{y}=\sigma\cup\{(i,y)\}$.}

    \emph{Proof.}    Towards a contradiction, suppose ownership at $\sigma^x$ and $\sigma^y$ would be as described above. Consider $R_j:zy...$ and $R_k:z...y$, so $y$ as least-preferred for $k$. Moreover, suppose that all agents unassigned at $\sigma$, besides $i$ and $j$, also consider $y$ least preferred. Then for $R_i:z...$, $i$ is assigned $z$ and $j$ is assigned $y$ by Pareto-efficiency. For $R'_i:x...$ we reach $\sigma^x$ and $j$ owns $z$ -- then $j$ is assigned $z$ as it is her most-preferred object. For $R''_i:y...$ we reach $\sigma^y$ and $k$ owns $z$ -- but then $k$ will be assigned $z$ so that $j$ is assigned neither $y$ nor $z$ but some third object. Thus, depending on $i$'s preferences, $j$ is assigned three different objects, a contradiction. \hfill{$\diamond$}

   \noindent \textbf{Claim 4}: \emph{If at some submatching $\sigma$, there is a single owner, then in the subsequent submatching there will again be a single owner. Moreover, that owner is the same at all subsequently reached submatchings, i.e., does not depend on the choice of the owner at $\sigma$. }
   
   \emph{Proof.}    Let $i$ be the sole owner at $\sigma$. We want to rule out that there are two owners $j$ and $k$ at the subsequent submatching that results from $i$'s choice.

    Towards a contradiction, suppose that $j$ owns $y$ and $k$ owns $z$ after $i$ picks $x$. Then, by claim 3, $j$ also owns $y$ if $i$ picks $z$. In the same way, $k$ also owns $z$ if $i$ picks $y$. Now, if $R_j:zy...$ and $R_k:yz...$, we get the following assignment of objects for $j$: if $i$ picks $x$, $j$ and $k$ trade and $j$ is assigned $z$. If $i$ picks $y$, $k$ owns $z$ and picks it so that $j$ is assigned neither $y$ nor $z$. If $i$ picks $z$, $j$ owns $y$ and picks it. Thus, depending on $i$'s choice, $j$ will be assigned at least three different objects, a contradiction. 
    
    Moreover, if there was a single owner, say $j$, after $i$ picks $x$ and another single owner, say $k\neq j$, after $i$ picks $y$, then there would be a third object $z$ whose  ownership moves between $j$ and $k$ depending on $i$'s choice -- again a violation of claim 3.  \hfill{$\diamond$}

    \noindent \textbf{Claim 5}: \emph{If at some submatching $\sigma$, there is a single owner $i$, and a sequence of agents  $j_1,j_2,...j_h$ who inherit objects in a fixed order following $i$'s choice, then $i$'s choice cannot influence who as subsequent owner inherits objects after $j_h$.}

     \emph{Proof.}    Consider the submatching arrived from $\sigma$ where $i$ picks $x\in \overline{O}_\sigma$ and each agent $j_1,...j_h$ picks their most preferred object according to their (fixed) preferences; denote it as $\sigma^x$. Now across different submatchings $\sigma^x$, $\sigma^y$... the corresponding sets of unassigned objects differ by one object: $|\cup_{o'\in \overline{O}_\sigma} \overline{O}_{\sigma^{o'}}|-1= |\overline{O}_{\sigma^{o}}|$, for all $o\in \overline{O}_\sigma$, see claim 1 in the proof of lemma \ref{lemma:SD-power}.  Now towards a contradiction, assume that for $\sigma^x$, $k$ is the owner while for $\sigma^y$, $l\neq k$ is the owner. Since there are (at least) two unassigned agents, there need to be two unassigned objects at $\sigma^x$ as well as at $\sigma^y$. Since the sets of unassigned objects are close to each other -- each differing by (at most) one element from the union -- there is an object included in both $\overline{O}_{\sigma^x}$ and $\overline{O}_{\sigma^y}$ that changes ownership; call it $z$. Then in the same way as in claim 3, we arrive at a contradiction. \hfill{$\diamond$}

    Thus, if at some point we reach a submatching with a single owner, then there is a sequence of subsequent submatchings with ownership determined in a fixed order, that cannot be affected by choices of (previous) owners. Next we show that even if initially there are two owners, then at any subsequent submaching there is only one owner.
   
   \noindent \textbf{Claim 6}: \emph{If there are two owners at $\sigma=\emptyset$, 
   there is at most one owner at the subsequent submatching.}

   \emph{Proof.}  Suppose $i$ owns $x$ and $j$ owns $y$ at $\sigma = \emptyset$. First, consider $\sigma^x=\{(i,x)\}$, i.e., a submatching reached where $i$  is assigned (one of) her owned object(s) and $j$ remains unassigned. (For example this would arise if $i$ considers $x$ most-preferred while $j$ considers one of the objects owned by $i$ most-preferred, either $x$ or some of $i$'s other objects.) At $\sigma^x$, $j$ still owns $y$. Towards a contradiction, suppose there is another agent $k$ who owns another object $z$. Since $j$'s ownership of objects is preserved, $z$ must have been owned by $i$ at $\sigma$. But then consider $R_k:zyx...$ and $R_j:z...xy$ (ranking $x$ second-to-last and $y$ last)  while every agent other than $i, j$ and $k$ also ranks $x$ second-to-last and $y$ last. For $R_i:x...$, $i$ picks $x$ at $\sigma$, $k$ owns $z$ at $\sigma^x$ and hence is assigned $z$. For $R'_i:y...$, $i$ and $j$ trade $z$ and $y$ at $\sigma$, so that $k$ is assigned $x$ by Pareto-efficiency (as he ranks it first among the remaining objects, while everyone else ranks it last). For $R''_i:z...$, $i$ picks $z$ at $\sigma$, then $k$ is assigned $y$ by Pareto-efficiency (ranks it first among remaining objects while everyone else ranks it last). Thus, depending on $i$'s preferences, $k$ would be assigned at least three different objects, a contradiction. Similarly there can be no two owners at a subsequent submatching $\sigma^y=\{(j,y)\}$ where $j$ is assigned one of the objects she owns.
   
   Next, consider $\sigma^{xy}=\{(i,x),(j,y)\}$, i.e., a submatching where both $i$ and $j$ are assigned one of the objects they own. Such a submatching is reached if both $i$ and $j$ consider one of their own objects most preferred, $x$ for $i$ and $y$ for $j$. However the same submatching is also reached for $R_i:yx...$ and $R_j:y...$, since then we first move to $\sigma^y$ (assigning $y$ to $j$); by the above we then know that there is only one owner at $\sigma^y$ and since $i$ is an owner we know that $i$ must own all remaining objects (not only $x$). Given $R_i:yx...$ we then reach $\sigma^{xy}=\{(i,x),(j,y)\}$ and by claim 4 we know that there can be only one owner at $\sigma^{xy}$ as this submatching is reached from $\sigma^y$.  
   
   Last, consider $\sigma^{yx}=\{(i,y),(j,x)\}$, i.e., a submatching where both $i$ and $j$ are assigned one of the objects owned by the other agent. Such a  submatching is reached if $i$ considers $y$ most preferred and $j$ considers $x$ most preferred. We want to show that any object $z$ is owned by the same agent in $\sigma^{yx}$ as in $\sigma^{xy}$ -- since we saw above that there is a single owner at $\sigma^{xy}$ this implies that there is a single owner at $\sigma^{yx}$. Towards a contradiction, assume that $k$ owns $z$ at $\sigma^{yx}$ while $l\neq k$ owns $z$ at $\sigma^{xy}$. 
   W.l.o.g., assume that $i$ owns $z$ (otherwise relabel $i$ and $j$). Further, suppose preferences are given as $R_j:xy...$, $R_k:zyx...$, $R_l:z...y$ while all other agents besides $i,j,k,$ and $l$ rank $y$ last. Then for $R_i:y...$ we reach $\sigma^{yx}$ and $k$ owns $z$; as it is her most-preferred object, she is assigned $z$. For $R_i':x...$ we first reach $\sigma^{x}$. 
   By the first case above, we know that there is a single owner at $\sigma^x$ and since $j$ still owns $y$ it must be $y$. After $j$ picks $y$ we then reach $\sigma^{xy}$ so that $l$ owns $z$ and is assigned $z$ -- but then $k$ is assigned neither $y$ nor $z$. Third, for $R''_i:z...$ $i$ picks $z$ (as its owner), $j$ becomes the owner of all remaining objects at $\sigma^{z}$ and picks $x$, and eventually, by Pareto-efficiency, $k$ is assigned~$y$ -- a contradiction, as $k$ is assigned three different objects depending on $i$'s preferences.   \hfill{$\diamond$}

Thus, in a setting with two owners, if only one of them is matched, then the other one becomes the owner, while if both are matched there is a single owner at the resulting submatching. Next we show that once both initial owners have been matched, the subsequent single owner must be the same across all possible submatchings, independent of the objects matched to the two initial owners.

\noindent \textbf{Claim 7}:  \emph{If there are two owners at $\sigma=\emptyset$, say $i$ and $j$, and at least three objects, say $x$, $y$, and $z$, then the owner of $z$ at $\sigma^{xy}=\{(i,x),(j,y)\}$ is the same as the owner of $y$ at $\sigma^{xz}$.}

    \emph{Proof.}  Suppose that $k$ owns $z$ at $\sigma^{xy}$, while $l\neq k$ owns $y$ at $\sigma^{xz}$. But then there are at least four agents and hence there must be another object $z'\neq x,y,z$. Then $k$ also owns $z'$ at $\sigma^{xy}$ while $l$ owns $z'$ at $\sigma^{xz}$.

    Suppose $i$ considers $x$ most-preferred. Then $\sigma^{xy}$ is reached if $j$ considers $y$ most preferred, while $\sigma^{xz}$ is reached if $j$ prefers $z$ over any other object. Note that given $i$'s preferences, $j$ can pick $y$, $z$ or $z'$ (by ranking it as most-preferred) just as the single owner can choose between three objects in claim 3. By the same argument as used in that claim, $j$ cannot influence who becomes the next owner, $k$ or $l$, of an object she doesn't pick -- in particular cannot affect the ownership of $z'$ by picking $y$ or $z$, a contradiction.  \hfill{$\diamond$} 

Since $i$, and $j$, as well as $x$, $y$, and $z$ were chosen arbitrarily in claim 7, we also know that the owner at $\sigma^{xy}$ is the same as at $\sigma^{zy}$, thus the same as at $\sigma^{zx}$, thus the same as at $\sigma^{yx}$. 

Thus we know that after $i$ and $j$, the two owners as $\sigma=\emptyset$, are matched, there is a unique owner $k$, independent of the objects to which $i$ and $j$ were matched. The last thing to show is the analogue of claim 5 for two owners. Let $i$ and $j$ be the two owners at $\sigma=\emptyset$. If $i$ considers one of her own objects, most-preferred, then, from the point of view of $j$, we are in a situation where $j$ is the owner of the remaining objects and we can reduce the analysis of $j$'s influence to the simple serial dictatorship scenario treated above. If instead, $i$ considers one of $j$'s objects most preferred, say $x$, and ranks $y$ second, then the subsequent owner, say $k\neq i,j$ will find that $x$ is already picked (either by $i$ or by $j$) and that one additional object has been assigned to $i$ or $j$ (either $j$'s most-preferred object if that differed from $x$, or $y$ if $j$ also considered $x$ most preferred). Thus, all potential sets of objects that $k$ may be able to choose from differ only by one element from the union of all these sets -- and as we show in claim 1 in the proof of Lemma 5, these potential sets that subsequent owners may choose from remain close to one another in that sense. But then the order of subsequent owners cannot change by the same argument as in claim 5. This completes the proof of Theorem \ref{thm:bipolar-SD-char}

\singlespacing
\setlength{\baselineskip}{1.1em} 
\bibliographystyle{apalike}
\bibliography{bibliography}
\end{document}